\documentclass[12pt]{amsart}

\usepackage{float}

\usepackage{graphicx}
\usepackage{url} 

\usepackage[a4paper,margin=3cm]{geometry}

\newtheorem{theorem}{Theorem}[section]
\newtheorem{lemma}[theorem]{Lemma}           
\newtheorem{cor}[theorem]{Corollary}

\theoremstyle{definition}
\newtheorem{definition}[theorem]{Definition}

\theoremstyle{remark}

\numberwithin{equation}{section}

\usepackage{color}


\title[Continuum limit for lattice Schr{\"o}dinger operators]
{Continuum limit for lattice Schr{\"o}dinger operators}
\author{Hiroshi ISOZAKI}
\address{Graduate School of Pure and Applied Sciences,
University of Tsukuba, 
Tsukuba, 305-8571, Japan}
\email{isozakih@math.tsukuba.ac.jp}
\author{Arne Jensen}
\address{Department of Mathematical Sciences, Aalborg University, Skjernvej 4A, 9220 Aalborg \O, Denmark}
\email{matarne@math.aau.dk}

\begin{document}
\baselineskip 15pt
\maketitle

\begin{abstract}
We study the behavior of  solutions of the Helmholtz equation 
$(- \Delta_{disc,h} - E)u_h = f_h$ on a periodic lattice as the mesh size $h$ tends to 0.  
Projecting to the eigenspace of a characteristic root $\lambda_h(\xi)$ and using a gauge transformation associated with the Dirac point, we show that the gauge transformed solution $u_h$ converges to that for the equation $(P(D_x) - E)v = g$ for a continuous model on ${\bf R}^d$, where $\lambda_h(\xi) \to P(\xi)$. For the case of the hexagonal and related lattices, {in a suitable energy region}, it converges to that for the Dirac equation. For the case of the square lattice,  triangular lattice, {hexagonal lattice (in another energy region)} and subdivision of a square lattice, one can add a scalar potential, and the solution of the lattice Schr{\"o}dinger  equation  $( - \Delta_{disc,h} +V_{disc,h} - E)u_h = f_h$  converges to that of the continuum Schr{\"o}dinger equation $(P(D_x) + V(x) -E)u = f$.
 \end{abstract}

\section{Introduction}
The lattice is a standard model to describe wave motions on periodic structures. The associated Laplacian $ \Delta_{\Gamma_h}$ is a difference operator. When the mesh size tends to 0, (a part of)  $H_{disc,h} = \frac{1}{h^{\nu}}\big(- \Delta_{\Gamma_h} - E_0\big)$ with a suitable scale factor $h^{\nu}$ and a reference energy $E_0$ has a formal limit $H_{cont}$
 as a (pseudo) differential operator $P(D_x)$, and one expects the convergence of solutions of the equation $(H_{disc,h} - E)u_h = f_h$ to those for the continuous model with Hamiltonian $H_{cont}$. The aim of this paper is to study this continuum limit of discrete periodic systems. We are mainly interested in solutions representing the scattering wave, i.e.  $u_{\pm,h} = (H_{disc,h}- E \mp i0)^{-1}f_h$, where $E \in \sigma_{cont}(H_{disc,h})$. We show that these scattering solutions of the lattice system converge to those for the continuous model as $h \to 0$, namely, {given a suitable relatively compact interval $I \subset {\bf R}$, }
\begin{equation}
\mathsf{J}_h(H_{disc,h} - E \mp i0)^{-1}{\mathsf P}_h \to (H_{cont}- E \mp i0)^{-1}{\mathsf P},
\label{IntroConvergence1}
\end{equation}
for all $E \in I$
in the strong sense in $L^{2,-s}$, $s > 1/2$, (see (\ref{DefineL2sRd})), 
 where ${\mathsf J}_h$ and  ${\mathsf P}_h$, $\mathsf P$ are suitable embedding and localization operators. This then yields
\begin{equation}
{\mathsf J}_he^{-itH_{disc,h}}{\mathsf E}_h(I){\mathsf P}_h \to e^{-itH_{cont}}{\mathsf P},
\label{IntroConvergence2}
\end{equation}
where ${\mathsf E}_h(\cdot)$ is 
is the spectral decomposition of $H_{disc,h}$.
Hence for any $\varphi \in C^{\infty}_0({\bf R})$, one can show the  convergence of the function of Hamiltonian:
\begin{equation}
{\mathsf J}_h\varphi(H_{disc,h}){\mathsf P}_h \to \varphi(H_{cont}){\mathsf P}.
\label{IntroConvergence3}
\end{equation}
Our method is also able to deal with a complex energy parameter $E$. In this case, one can derive (\ref{IntroConvergence3}), hence (\ref{IntroConvergence2}), even if one cannot show (\ref{IntroConvergence1}). On some lattices, one can add a potential to $H_{disc,h}$, in which case,  one can argue the convergence of observables in scattering phenomena, e.g. the S-matrix.

We pick up a characteristic root of $\frac{1}{h^{\nu}}\big( - \Delta_{\Gamma_h}- E_0\big)$ and pass to the gauge transformation to derive the convergence of characteristic root $\lambda_h(\xi)$ to $P(\xi)$. By projecting to the associated eigenspace, one derives the desired convergence. 
  The gauge transformation is inspired by the expansion around Dirac points for the hexagonal lattice, hence our method covers carbonic lattices like graphene, graphite and the Kagome lattice. In this case, the associated continuous system is the two-dimensional massless Dirac operator. Discrete systems related to the square lattice, e.g. ladders or subdivisions, are also dealt with. In particular, for the case of square and triangular lattices, also for ladders, one can add compactly supported potentials. 

Given an $h$-independent lattice Hamiltonian $\mathcal L(S)$, where $S = (S_1,\cdots,S_d)$ is a shift operator (see \S \ref{S2Preliminaries}), one first chooses a reference energy $E_0$, and consider the scaled Hamiltonian
$$
\mathcal L_h(S_h) = \frac{1}{h^{\nu}}\big(\mathcal L(S_h) - E_0\big).
$$
One should note that the scaling order $\nu$ depends on the energy region. For example, in the case of the hexagonal lattice, 
$\nu$'s are different near the middle of the spectrum and  near the end points of the spectrum. This is due to the behavior of characteristic roots near the local extremal points.

The proof is based on the compactness argument in elementary topology: A precompact sequence $\{y_i\}$ in a complete metric space $Y$ having a unique accumulation point is convergent in $Y$. This basic argument has been used very often in the study of the continuous spectrum of Schr{\"o}dinger operators in $L^2({\bf R}^d)$.
Let $H = - \Delta + V(x)$, and put $R(z) = (H - z)^{-1}$. To study the continuous spectrum of $H$, a first important step is the limiting absorption principle (LAP), i.e. the existence of the limit 
\begin{equation}
R(E \pm i0) = \lim_{\epsilon\to0}R(E \pm i\epsilon) : X \to Y, \quad E \in \sigma_c(H)
\nonumber
\end{equation}
for suitable Banach spaces $X, Y$ rigging $L^2({\bf R}^d)$, i.e. $X \subset L^2({\bf R}^d) \subset Y$. The classical work of Eidus \cite{Eidus} proved the LAP by using the above compactness argument, and the uniqueness of solutions for Schr{\"o}dinger equations satisfying the radiation condition played an important role. 
The LAP has been extended to more general differential operators by Agmon \cite{Agmon75}, Kato-Kuroda \cite{KatoKuroda71}, \cite{Kuroda}, J{\"a}ger \cite{Jager}, Ikebe-Saito \cite{IkebeSaito}, Agmon-H{\"o}rmander \cite{AgHo76} by using Fourier analysis,  abstract operator theory, or integration by parts machinery. The commutator method of  Mourre \cite{Mourre} is apparently different, however, it can be rewritten into the above mentioned form.
The LAP is also valid  for discrete Schr{\"o}dinger operators. We can derive uniform estimates with respect to $0 < h < h_0$ of $u_h$ and the radiation condition for the discrete equation \cite{AndIsoMor1}.
We shall use this argument also in the passage from  discrete  to continuous. We define $\widetilde u_h(x)$ by (\ref{DefineTdh}). 
Using  uniform estimates, we can show that $\{\widetilde u_h(x)\}$ has the unique accumulation point as $h \to 0$, which guarantees that $\widetilde u_h(x)$ itself converges to the unique solution $\widetilde u$ to $(P(D_x) - E)\widetilde u = f$.

In \S 2, we review basic assumptions for the lattices studied in 
\cite{AndIsoMor1}. In \S 3, we summarize the conditions on the 
 characteristic roots of a discrete system
needed for the passage to continuous system. In \S 4, 
we study the free system (the case without potential) in a general form. The main results are Theorems \ref{TheoremwidetildeuhtowidetildeufordiscreteCase1} and  \ref{TheoremwidetildeuhtowidetildeufordiscreteCase2}.  In \S 5, we study the case with potential also in a general form. In \S 6, we study the complex energy case. In  \S 7 we apply these theorems to the square and triangular lattices, and also to the ladder and subdivision of square lattices. For these cases, one can add a potential, hence the solutions to lattice Schr{\"o}dinger equations converge to those for the continuum Schr{\"o}dinger operator $- \Delta + V(x)$. We then prove that the S-matrix for the continuum model is approximated by that for the discrete model.  
 In  \S 8, we consider the hexagonal lattice, the graphite lattice and the Kagome lattice. Here, we derive the Dirac equation as well as the Schr{\"o}dinger equation. Some technical lemmas are proved in the Appendix. 

Ignat-Zuazua \cite{IgZua09} proved that for the time-dependent Schr{\"o}dinger equation, the usual approximation scheme does not converge in some topologies and proposed a new approximation scheme. This scheme they then used to obtain fundamental estimates used in the study of the continuum limit of non-linear Schr\"{o}dinger equations. The continuum limit of non-linear discrete Schr\"{o}dinger operators was studied by Hong and Yang~\cite{HoYa19} employing mesh-uniform Strichartz estimates. See this paper for further references to the continuum limit of non-linear discrete Schr\"{o}dinger operators.   Nakamura-Tadano \cite{NakaTada} proved the norm convergence of resolvents of Schr{\"o}dinger operators for complex energies. 
Dirac operators are often used to study the spectral structure of graphene. See e.g. \cite{GueSieden}, \cite{BarCornStock}, \cite{BarCornZalc}. The resolvent of a discrete Schr\"{o}dinger operator may have singularities at interior points of the continuous spectrum. For the discrete Laplacian on the square lattice the singularity structure was investigated in \cite{Ito-Jensen}.

We use a standard notation. 
For $a \in {\bf R}$, $[a]$ denotes the greatest integer not exceeding $a$. For Banach spaces $X$, $Y$, ${\bf B}(X,Y)$ is the space of all bounded operators from $X$ to $Y$. For $x \in {\bf R}^d$, $\langle x\rangle = (1 + |x|^2)^{1/2}$, where $|x| = (\sum_{i=1}^dx_i^2)^{1/2}$ for $x = (x_1,...,x_d) \in {\bf R}^d$.
 For an interval $I \subset {\bf R}$ and a Hilbert space ${\bf h}$, $L^2(I,{\bf h},\rho(t)dt)$ denotes the set of ${\bf h}$-valued $L^2$-functions on $I$ with respect to the measure $\rho(t)dt$. The notation $P(x,D_x)$ denotes the pseudo-differential operator
$$
P(x,D_x) u(x) = (2\pi)^{-d/2}\int_{{\bf R}^d}e^{ix\cdot\xi}P(x,\xi)
(\mathcal F_{cont}u)(\xi)d\xi,
$$
where $\mathcal F_{cont}$ is defined in (\ref{DefmathcalF}).


\section{Preliminaries}
\label{S2Preliminaries}


\subsection{Fourier transform}
We consider the square lattice with mesh size $h$, i.e. 
\begin{equation}
{\bf Z}^d_h = \{hn\, ; \, n \in {\bf Z}^d\}.
\nonumber
\end{equation}
We equip $L^2({\bf Z}^d_h)$ with norm
\begin{equation}
\|a\|_{L^2({\bf Z}^d_h)} = h^{d/2}\Big(\sum_{n \in {\bf Z}^d}|a_n|^2\Big)^{1/2}.
\label{DefineL2Zdhnorm}
\end{equation}
Put
\begin{equation}
I_{hn} = hn + \Big[-\frac{h}{2},\frac{h}{2}\Big]^d.
\nonumber
\end{equation}
For $a = \big(a_n\big) \in L^2({\bf Z}^d_h)$, we define $f_a(x) \in L^2({\bf R}^d)$ by
\begin{equation}
f_a(x) = a_n, \quad {\rm if} \quad x \in I_{hn}.
\nonumber
\end{equation}
Then by the definition (\ref{DefineL2Zdhnorm}) we have
\begin{equation}
\int_{{\bf R}^d}|f_a(x)|^2dx = \|a\|^2_{L^2({\bf Z}^d_h)}.
\nonumber
\end{equation}
In the following $C$'s denote constants independent of $0 < h < h_0$, where $h_0 > 0$ is a suitable fixed small constant.

Let $\mathcal S({\bf Z}_h^d)$ be the space of rapidly decreasing sequences on ${\bf Z}_h^d$:
\begin{equation}
\mathcal S({\bf Z}^d_h) \ni a = (a_n)_{n \in {\bf Z}^d} 
\Longleftrightarrow 
|a_n| \leq C_k \langle n\rangle^{-k}, \quad \forall n \in {\bf Z}^d, \quad \forall k \geq 0.
\nonumber
\end{equation}
Its dual space is denoted by $\mathcal S'({\bf Z}^d_h)$. It is embedded in $\mathcal S'({\bf R}^d)$ by
\begin{equation}
 \mathcal S'({\bf Z}^d_h) \ni a \to 
\sum_{n \in {\bf Z}^d}a_n\delta (x - hn) \in \mathcal S'({\bf R}^d),
\nonumber
\end{equation}
where $\delta(x)$ denotes the Dirac measure supported at $0 \in {\bf R}^d$.
Let ${\bf  T}^d_h$ be the $d$-dimensional torus of size $2\pi/h$:
\begin{equation}
{\bf T}^d_{h} = \big(S^1_h\big)^d = \Big[-\frac{\pi}{h},\frac{\pi}{h}\Big]^d, \quad S^1_h = \Big\{e^{ih\theta}\, ; 
\, -\frac{\pi}{h} \leq \theta  \leq \frac{\pi}{h}\Big\}.
\nonumber
\end{equation}
We  use the following three symbols to denote functions on ${\bf Z}^d_h$, ${\bf T}^d_h$ and ${\bf R}^d$:
\begin{equation}
u_h = \big(u_h(n)\big) \quad {\rm on} \quad {\bf Z}^d_h,
\label{Defineuh(n)}
\end{equation}
\begin{equation}
\widehat u_h(\xi) = \big(\dfrac{h}{2\pi}\big)^{d/2}\sum_{n\in{\bf Z}^d}e^{-ihn\cdot\xi}u_h(n) \quad {\rm on} \quad {\bf T}^d_h,
\label{Definewidehatuh(xi)}
\end{equation}
\begin{equation}
\widetilde u_h(x) = \big(\dfrac{h}{2\pi}\big)^{d/2}\int_{{\bf T}^d_h}e^{ix\cdot\xi}\widehat u_h(\xi)d\xi \quad {\rm on} \quad{\bf R}^d.
\label{DefineTdh}
\end{equation}
For $f \in \mathcal S({\bf R}^d)$ consider the series of transformations:
$$
f \Longrightarrow f_h \Longrightarrow \widehat f_h \Longrightarrow \widetilde f_h,
$$
where $f_h(n) = f(hn)$. 
Then $\widetilde f_h$ is an interpolation of $f_h$, and $\widetilde f_h \to f$ as $h \to 0$. In fact, since
$$
\widetilde f_h(x) = (2\pi)^{-d}\int_{{\bf T}^d_h}
\Big(\sum_ne^{-ihn\cdot\xi}f(hn)h^d\Big)d\xi,
$$
it converges to 
$$
(2\pi)^{-d/2}\int_{{\bf R}^d}e^{ix\cdot\xi}\big(\mathcal F_{cont}f\big)(\xi)d\xi
$$
in the sense of distribution, where 
 $\mathcal F_{cont}$ is the Fourier transformation on ${\bf R}^d$:
\begin{equation}
(\mathcal F_{cont}f)(\xi) =  (2\pi)^{-d/2}\int_{{\bf R}^d}e^{-ix\cdot\xi}f(x)dx.
\label{DefmathcalF}
\end{equation}
Note that 
\begin{equation}
\widetilde f_h(x) = \sum_{n\in{\bf Z}^d}\left(\prod_{j=1}^d\frac{\sin \frac{\pi}{h}(x_j - hn_j)}{\frac{\pi}{h}(x_j - hn_j)}\right)f(hn),
\nonumber
\end{equation}
and the right-hand side is called the cardinal series (see \cite{Schoenberg}). 

Define the discrete Fourier transform of $a_h = (a_{h,n}) \in \mathcal S({\bf Z}^d_h)$ by
\begin{equation}
\big(\mathcal F_{disc,h} a_h\big) (\xi) = \widehat a_h(\xi) = \big(\frac{h}{2\pi}\big)^{d/2}\sum_{n \in {\bf Z}^d}a_{h,n}e^{-ihn\cdot\xi}.
\label{DiscreteFourierTransf}
\end{equation}
The inverse Fourier transform is
\begin{equation}
a_h = \big(a_{h,n}\big) = (\mathcal F_{disc,h})^{-1}\widehat a_h, \quad  a_{h,n} = \big(\frac{h}{2\pi}\big)^{d/2}\int_{{\bf T}^d_h}e^{ihn\cdot\xi}\widehat a_h(\xi)d\xi
\nonumber
\end{equation}
Since $\Big\{\big(\dfrac{h}{2\pi}\big)^{d/2} e^{-ihn\cdot\xi}\, ; \ n \in {\bf Z}^d\Big\}$ is an orthonormal basis of $L^2({\bf T}^d_h)$, we have the following lemma.

\begin{lemma} The discrete Fourier transform
\begin{equation}
\mathcal F_{disc,h} : L^2({\bf Z}^d_h) \ni a_h \to  \widehat a_h = \big(\frac{h}{2\pi}\big)^{d/2}\sum_{n \in {\bf Z}^d}a_{h,n}e^{-ihn\cdot\xi} 
\in L^2({\bf T}^d_h)
\label{DiscreteFouriertransf}
\end{equation}
is a bijection, in particular, it is isometric in the following sense:
\begin{equation}
h^{-d}\|a_h\|^2_{L^2({\bf Z}^d_h)} = \sum_{n \in {\bf Z}^d}|a_{h,n}|^2 = \|\widehat a_h\|^2_{L^2({\bf T}^d_h)} =  
h^{-d}\|\widetilde a_h\|^2_{L^2({\bf R}^d)}.
\label{h-ddL2=suman2=intaxidxi}
\end{equation}
\end{lemma}
In particular, we have
\begin{equation}
\|a_h\|_{L^2({\bf Z}^d_h)} = \|\widetilde a_h\|_{L^2({\bf R}^d)}.
\label{S2amnorm=widetildeahnorm}
\end{equation}

The shift operators $S_{h,j}$ on $\mathcal S'({\bf R}^d)$ and  $\mathcal S'({\bf Z}^d_h)$ are defined by
\begin{equation}
\big(S_{h,j} f\big)(x) = f(x - h{\bf e}_j),
\label{S2SchiftOP1}
\end{equation}
\begin{equation}
\big(S_{h,j}u_h\big)(n) = u_h(n - {\bf e}_j),
\label{S2ShiftOP2}
\end{equation}
where
\begin{equation}
{\bf e}_1 = (1,0,\cdots,0),  \cdots,  {\bf e}_d = (0, \cdots, 0,1)
\nonumber
\end{equation}
is the standard basis of ${\bf R}^d$.
We put 
\begin{equation}
S_h = (S_{h,1},\cdots,S_{h,d}).
\label{S2ShiftOP3}
\end{equation}
Then we have
\begin{equation}
\mathcal F_{cont} S_{h,j} = e^{-ih\xi_j}\mathcal F_{cont},
\nonumber
\end{equation}
\begin{equation}
\mathcal F_{disc,h} S_{h,j} = e^{-ih\xi_j}\mathcal F_{disc,h}.
\nonumber
\end{equation}

\subsection{Function spaces}
First we recall the Agmon-H{\"o}rmander space (\cite{AgHo76}, \cite{HoVol2}), which is a Besov space. Let
\begin{equation}
\Omega_0 = \{|x| < 1\}, \quad 
\Omega_j = \{2^{j-1} < |x| < 2^j\}, \quad j \geq 1, 
\nonumber
\end{equation}
\begin{equation}
 \mathcal B({\bf R}^d) \ni u \Longleftrightarrow \sum_{j=0}^{\infty}2^{j/2}\|u\|_{L^2(\Omega_j)} < \infty, 
 \nonumber
\end{equation}
\begin{equation}
\mathcal B^{\ast}({\bf R}^d) \ni u \Longleftrightarrow \sup_{R>1}\frac{1}{R}\int_{|x|<R}|u(x)|^2 dx < \infty.
\nonumber
\end{equation}
For $f, g \in {\mathcal B}({\bf R}^d)$, we define
\begin{equation}
f \simeq g \Longleftrightarrow \lim_{R\to\infty}
\frac{1}{R}\int_{|x|<R}|\mathcal F_{cont}f(\xi) - \mathcal F_{cont}g(\xi)|^2d\xi = 0.
\nonumber
\end{equation}
Finally, we define for $s \in {\bf R}$
\begin{equation}
L^{2,s}({\bf R}^d) \ni u \Longleftrightarrow \langle x\rangle^su \in L^2({\bf R}^d),
\label{DefineL2sRd}
\end{equation}
\begin{equation}
H^{m}({\bf R}^d) \ni u \Longleftrightarrow  \langle\xi\rangle^m\mathcal F_{cont}u(\xi) \in L^2({\bf R}^d), 
\nonumber
\end{equation}
\begin{equation}
H^{m,s}({\bf R}^d) \ni u \Longleftrightarrow \langle x\rangle^s u \in H^m({\bf R}^d).
\label{S2DefineHms}
\end{equation}

We consider the discrete analogues of these spaces.
For $s \in {\bf R}$, we define
\begin{equation}
L^{2,s}({\bf Z}^d_h) \ni u \Longleftrightarrow \|u\|_{L^{2,s}({\bf Z}^d_h)} = \Bigl(h^d\sum_{n \in {\bf Z}^d}\langle hn\rangle^{2s}|u(n)|^2\Bigr)^{1/2}.
\nonumber
\end{equation}
Define 
\begin{eqnarray}
\Omega_{h,0}
&=&\{n \in {\bf Z}^d\, ; \, |hn| \leq 1\}, \nonumber \\
\Omega_{h,\ell} &=& \{n \in {\bf Z}^d\, ; \, 2^{\ell-1} < |hn| \leq 2^{\ell}\}, \quad \ell \geq 1.
\nonumber
\end{eqnarray}
Define $\mathcal B({\bf Z}^d_h)$ by
\begin{equation}
\mathcal B({\bf Z}^d_h) \ni u(n) \Longleftrightarrow 
\sum_{\ell = 0}^{\infty}2^{\ell/2}\|u\|_{L^2(\Omega_{h,\ell})}< \infty,
\nonumber
\end{equation}
where
$
\|u\|_{\Omega_{h,\ell}} = \Big(h^d\sum_{n \in \Omega_{h,\ell}}|u(n)|^2\Big)^{1/2}.
$
{The norm of the dual spece of $\mathcal B({\bf Z}^d_h)$ should be $\sup_{\ell \geq 0}2^{-\ell/2}\|u\|_{L^2(\Omega_{h,\ell})}$. However, 
by Lemmas \ref{alphabetalemma} and \ref{AB2ALemma}, there exists a constant $C > 0$ independent of $0 < h < 1$ such that
$$
C\Big(
\sup_{R>1} \frac{h^d}{R}\sum_{|hn| \leq R}|u(n)|^2
\Big)^{1/2} \leq \sup_{\ell \geq 0}2^{-\ell/2}\|u\|_{L^2(\Omega_{h,\ell})} 
\leq C^{-1}\Big(
\sup_{R>1} \frac{h^d}{R}\sum_{|hn| \leq R}|u(n)|^2
\Big)^{1/2}.
$$
Therefore we employ
\begin{equation}
\|u\|_{{\mathcal B}_h^{\ast}({\bf Z}^d)}= 
\Bigl(
\sup_{R>1} \frac{h^d}{R}\sum_{|hn| \leq R}|u(n)|^2
\Bigr)^{1/2}
\nonumber
\end{equation}
as the norm of ${\mathcal B}^{\ast}({\bf Z}^d_h)$.}
We define
\begin{equation}
\mathcal B({\bf T}^d_h) \ni u \Longleftrightarrow \big(\mathcal F_{disc,h}\big)^{-1}u \in {\mathcal B}({\bf Z}^d_h),
\nonumber
\end{equation}
\begin{equation}
\mathcal B^{\ast}({\bf T}^d_h) \ni u \Longleftrightarrow \big(\mathcal F_{disc,h}\big)^{-1}u \in {\mathcal B}^{\ast}({\bf Z}^d_h).
\nonumber
\end{equation}

An invariant way of defining the Besov spaces $\mathcal B$ and $\mathcal B^{\ast}$ on ${\bf T}^d_h$ is to use the Laplacian $-\Delta_{\xi} = -\sum_{i=1}^d\big(\partial/\partial \xi_i\big)^2$ on ${\bf T}^d_h$, which has eigenvalues $(hn)^2$ and eigenvectors  $\big(\frac{h}{2\pi}\big)^{d/2} e^{-ihn\cdot\xi}$, $n \in {\bf Z}^d$. We then define
\begin{equation}
\mathcal B({\bf T}^d_h) \ni u(\xi) \Longleftrightarrow 
\sum_{\ell \geq 0}2^{\ell/2}h^{d/2}\|\chi_{\ell}(\sqrt{- \Delta_{\xi}})u\| < \infty,
\nonumber
\end{equation}
where $\chi_{\ell}$ is the characteristic function of the interval $(c_{\ell-1},c_{\ell})$,
$
c_{-1}=0, \ c_{\ell} = 2^{\ell}, \ \ell \geq 0,
$
and 
\begin{equation}
\mathcal B^{\ast}({\bf T}^d_h) \ni u(\xi) \Longleftrightarrow 
\sup_{R>1}\frac{h^d}{R}\|\chi_R(\sqrt{-\Delta_{\xi}})u\|^2 < \infty,
\nonumber
\end{equation}
where $\chi_{R}$ is the characteristic functions of the interval $(0,R)$.

\subsection{Lattice Schr{\"o}dinger operators}
We introduce the mesh size parameter $h > 0$ in the definition of the lattice in \cite{AndIsoMor1}. 
Let ${\bf L}_{h}$ be a lattice of rank $d$ in ${\bf R}^d$ $(d \geq 2)$ with basis ${\bf v}_j, j = 1, \cdots, d$, and mesh size $h$, i.e.
$$
{\bf L}_{h} = \big\{{\bf v}(hn)\, ; \, n \in {\bf Z}^d\big\}, \quad 
{\bf v}(hn) = \sum_{j=1}^dhn_j {\bf v}_j, \quad 
n = (n_1,\cdots,n_d) \in {\bf Z}^d.
$$
For $h=1$, ${\bf L}_{h}$ is denoted by ${\bf L}$. 
Take points $p_j \in {\bf R}^d$, $j = 1, \cdots, \mathsf s$, satisfying
$$
p_i - p_j \not\in {\bf L}, \quad {\rm if} \quad i \neq j,
$$
and define the vertex set ${\bf V}_h$ by
\begin{equation}
{\bf V}_h = {\mathop\bigcup_{j=1}^{\mathsf s}}\big(hp_j + {\bf L}_h\big).
\nonumber
\end{equation}
For $h =1$, ${\bf V}_1$ is denoted by ${\bf V}$.
There exists a bijection ${\bf V} \ni a \to (j(a),n(a)) \in \{1,\cdots,{\mathsf s}\} \times {\bf Z}^d$ such that 
$$
a = p_{j(a)} + {\bf v}(n(a)).
$$
The group ${\bf Z}^d$ acts on ${\bf V}_{h}$ as follows:
$$
{\bf Z}^d\times {\bf V}_{h} \ni (m,a) \to m\odot a := hp_{j(a)} + h{\bf v}(m + n(a)) \in {\bf V}_{h}.
$$
The edge set ${\bf E}_{h}$ is a subset of ${\bf L}_{h}\times {\bf L}_{h}$  having the property
$$
{\bf E}_{h} \ni (a,b) \to (m\odot a, m\odot b) \in {\bf E}_{h}, \quad \forall m \in {\bf Z}^d.
$$
Then the triple ${\bf \Gamma}_{h} = \{{\bf  L}_{h}, {\bf V}_{h},{\bf  E}_{h}\}$ 
is a periodic graph in ${\bf R}^d$ with mesh size $h$. 
As above,  for $h = 1$,  $\Gamma_h$ is denoted by $\Gamma = \{{\bf L}, {\bf V}, {\bf E}\}$.
For $a, b \in {\bf V}_{h}$, $a \sim b$ means that they are adjacent, i.e. they are on the mutually opposite end points of an edge in ${\bf E}_{h}$, and ${\rm deg}\,(v)$ of $v \in {\bf V}_{h}$ is the number of vertices adjacent to $v$.
Then ${\rm deg}(p_i + {\bf v}(n))$ depends only on $i$, which is denoted by  ${\rm deg}(i)$. In the following we assume that
\begin{equation}
 {\rm deg}(i) \  do es\ not \ depend \ on \ i,  \ and  \ is  \ denoted \  by \  d_g.
\nonumber
\end{equation}
 
 Any function $f$ on ${\bf V}_{h}$ is written as $f(n) = (f_1(n),\cdots,f_{\mathsf s}(n))$, $n \in {\bf Z}^d$, where $ f_j(n)$ is identified with a function on $p_j + \mathcal L_0$.  Hence, $L^2({\bf V}_{h})$ is the Hilbert space with inner product
$$
(f,g)_{L^2({\bf V}_{h})} = 
\sum_{j=1}^{\mathsf s}(f_j,g_j)_{d_g},
$$
where 
\begin{equation}
(f_j,g_j)_{d_g} = d_g(f_j,g_j)_{L^2({\bf L})}.
\label{Defineinnerproductwithdeg}
\end{equation}

The Laplacian $\Delta_{{\bf \Gamma}_{h}}$ on ${\bf \Gamma}_{h}$ is defined by the mean over the adjacent vertices
\begin{equation}
\big(\Delta_{{\bf \Gamma}_{h}}f\big)_j(n) = \frac{1}{d_g}\sum_{b\sim p_j + {\bf v}(n)}
f_{i(b)}(n(b)).
\label{DefineDeltaGamma0hj}
\end{equation}
We then define a unitary operator $\mathcal U_{\bf\Gamma_{h}}\colon L^2({\bf V}_{h}) \to 
L^2({\bf T}^d_h)^{\mathsf s}$ by
\begin{equation}
\big(\mathcal U_{{\bf \Gamma}_{h}}f)_j = \sqrt{d_g}\mathcal F_{disc,h}f_j,
\nonumber
\end{equation}
where $L^2({\bf T}^d_h)^{\mathsf s}$ is equipped with the inner product
$$
(f,g)_{L^2({\bf T}^d_h)} = \sum_{j=1}^{\mathsf s}\int_{{\bf T}^d_h}f_j(\xi)\overline{g_j(\xi)}d\xi.
$$
Note that the Laplacian is written as
\begin{equation}
- \Delta_{{\bf \Gamma}_{h}} = {\mathcal L}(S_h,S_h^{\ast}),
\nonumber
\end{equation}
where ${\mathcal L}(z,w)$ is a matrix whose entries are polynomials in $z, w \in {\bf C}^d$.
Passing to the Fourier series, it is transferred to 
\begin{equation}
\mathcal U_{{\bf \Gamma}_{h}}{\mathcal L}(S_h,S_h^{\ast})\big(\mathcal U_{{\bf \Gamma}_{h}}\big)^{-1}
= {\mathcal L}(e^{-ih\xi},e^{ih\xi}),
\label{DefineWidehatDh}
\end{equation}
where we use the notation 
$$
e^{\pm ih\xi} = 
(e^{\pm ih\xi_1},\cdots,e^{\pm ih\xi_d}).
$$
 Then the spectrum of the operator $H_{h} = - \Delta_{{\bf \Gamma}_{h}}$ is equal to that
of the operator of multiplication by the function ${\mathcal  L}(e^{-ih\xi},e^{ih\xi})$ on $L^2({\bf T}^d_h)^{\mathsf s}$. 
Let us call ${\mathcal L}(e^{-ih\xi},e^{ih\xi})$ the symbol of $- \Delta_{{\bf \Gamma}_h}$. The characteristic root of ${\mathcal L}(e^{-ih\xi},e^{ih\xi})$ is a root $\lambda = \lambda_h(\xi)$ of the equation $\det\big(\mathcal L(e^{-ih\xi},e^{ih\xi}) - \lambda\big)=0$, and the characteristic surface is the set  $\{\xi \,; \, \det\big(\mathcal L(e^{-ih\xi},e^{ih\xi}) - \lambda\big)=0\}$.

We recall the assumptions on the lattice
in \cite{AndIsoMor1}. We have only to state them for the case $h = 1$, i.e. for $\Gamma$. Let 
 $\lambda_1\xi) \leq \lambda_2(\xi) \leq \cdots \leq \lambda_{\mathsf{s}}(\xi)$ denote the eigenvalues of ${\mathcal L}(e^{-i\xi},e^{i\xi})$.
  We put
\begin{equation}
{\bf T}^d = {\bf R}^d/(2\pi{\bf Z})^d, \quad 
{\bf T}^d_{\bf C} = {\bf C}^d/(2\pi{\bf Z})^d,
\nonumber
\end{equation}
\begin{equation}
p(\xi,\lambda) = \det\big({\mathcal L}(e^{-i\xi},e^{i\xi}) - \lambda\big).
\nonumber
\end{equation}
\begin{equation}
M_{\lambda} = \big\{\xi \in {\bf T}^d\, ; \, p(\xi,\lambda) = 0\},
\nonumber
\end{equation}
\begin{equation}
M_{\lambda,j} = \big\{\xi \in {\bf T}^d\, ; \, \lambda_j(\xi) = \lambda\},
\nonumber
\end{equation}
\begin{equation}
M_{\lambda}^{{\bf C}} = \big\{z \in {\bf T}^d_{\bf C}\, ; \, p(z,\lambda) = 0\},
\nonumber
\end{equation}
\begin{equation}
M_{\lambda,reg}^{\bf C} = \big\{z \in M^{\bf C}_{\lambda}\, ; \, \nabla_zp(z,\lambda) \neq 0\big\},
\nonumber
\end{equation}
\begin{equation}
M_{\lambda,sng}^{\bf C} = \big\{z \in M^{\bf C}_{\lambda}\, ; \, \nabla_zp(z,\lambda) = 0\big\},
\nonumber
\end{equation}
\begin{equation}
\widetilde{\mathcal T} = \big\{\lambda \in \sigma(H)\, ; \, M^{\bf C}_{\lambda,sng}\cap {\bf T}^d \neq 
\emptyset\big\}.
\nonumber
\end{equation}
The assumptions are as follows. Let $H = - \Delta_{{\bf \Gamma}}$.

\medskip
\noindent
({\bf A-1}) {\it There exists a subset ${\mathcal T}_1 \subset \sigma (H)$ such that for $\lambda \in \sigma(H)\setminus \mathcal T_1$:

\smallskip
({\bf A-1-1}) $M_{\lambda,sng}^{\bf C}$ is discrete.

\smallskip
({\bf A-1-2}) Each connected component of $M^{\bf C}_{\lambda,reg}$ intersects with ${\bf T}^d$ and the intersection is a $(d-1)$-dimensional analytic submanifold of ${\bf T}^d$.}

\medskip
\noindent
{\it 
({\bf A-2}) There exists a finite set $\mathcal T_0 \subset \sigma(H)$ such that}
$$
M_{\lambda,i}\cap M_{\lambda,j} = \emptyset, \quad {if} \quad i \neq j, \quad \lambda \in \sigma(H_0)\setminus\mathcal T_0.
$$
({\bf A-3}) $\ \nabla_{\xi}p(\xi,\lambda) \neq 0$ on $M_{\lambda}$ for $\lambda \in \sigma(H)\setminus \mathcal T_0$.

\medskip
\noindent
({\bf A-4}) {\it The unique continuation property holds for $\Delta_{{\bf \Gamma}}$ in ${\bf V}$.
}

\medskip
By the unique continuation property we mean the following assertion:
 
\smallskip
{\it Assume $u$ satisfies $(- \Delta_{{\bf \Gamma}}- \lambda)u = 0$ on ${\bf V}$ for some constant $\lambda \in {\bf C}$. If there exists $R_0 > 0$ such that $u = 0$ for $|n| > R$, then $u = 0$ on $\bf V$.
}

\smallskip

The assumption on the potential is as follows:

\medskip
\noindent
({\bf V-1}) {\it The potential $V_{disc,h}$ is a scalar multiplication operator defined by 
$$
\big(V_{disc,h}u\big)(n) = V(hn)u(n),
$$
where  $V(x)$ is a real-valued compactly supported continuous function on ${\bf R}^d$}.

\medskip
Under these assumptions, for any $0 < h < h_0$, we have the limiting absorption principle,  spectral representation, completeness of wave operators, unitarity of the S-matrix for $- \Delta_{{\bf \Gamma}_{h}} + V_{disc,h}$. Moreover, we can solve  inverse scattering problems (see \cite{AndIsoMor2}).


\section{Expansion at local extremal points}

\subsection{Local extremals}
In the sequel we use results from perturbation theory for matrices depending on several complex parameters. See~\cite{Ba} for these results.

For a subset $K \subset {\bf R}^d$, we put 
 $$
 K/h = \{\xi/h\, ; \, \xi \in K\}.
 $$
 
Let $\Delta_{{\bf \Gamma}}$ be  the Laplacian in the previous section, 
and take $E_0 \in \widetilde{\mathcal T}$. Then at $\lambda = E_0$, 
$p(\xi,\lambda) = \det\big({\mathcal L}(e^{-i\xi},e^{i\xi}) - \lambda\big)$ may have multiple roots or, some simple characteristic root $\lambda_j(\xi)$ satisfies $\nabla_{\xi}\lambda_{j}(\xi) = 0$ at some point in the characteristic surface $\{\xi\,  ; \lambda_j(\xi) = E_0\}$.
We consider the behavior of solutions of the equation 
$\big(- \frac{1}{h^{\nu}}\Delta_{{\bf \Gamma}_h} + V_{disc,h} - \lambda)u_h = f_h$ as $h \to 0$, when $\lambda$ is close to $E_0/h^{\nu}$.  We call $E_0$ the {\it reference energy}. 
For the sake of simplicity of notation, we shift the Hamiltonian and consider
\begin{equation}
\Delta_{disc,h} = \frac{1}{h^{\nu}}\big(\Delta_{{\bf\Gamma}_h} + E_0)
\label{DefineDeltadisch}
\end{equation}
instead of $\Delta_{{\bf \Gamma}_h}$. 
Then the symbol of $- \Delta_{disc,h} $ has the following form
\begin{equation}
\mathcal L_h(z,w) = \frac{1}{h^{\nu}}\Big(\mathcal L(z,w) - E_0I\Big),
\label{Lh=hnuLzw}
\end{equation}
$I$ being the ${\mathsf s} \times {\mathsf s}$ identity matrix.

We use the abbreviation :
 $$
 {\mathcal L}_{h}(e^{-ih\eta}) := {\mathcal  L}_h(e^{-ih\eta},e^{ih\eta}).
 $$
 For $\eta \in {\bf R}^d$ let 
 $$
 \lambda_{1,h}(\eta) \leq \cdots \leq \lambda_{{\mathsf s},h}(\eta)
 $$
 be the characteristic roots of ${\mathcal  L}_h(e^{-ih\eta},e^{ih\eta})$.  
 They are rewritten as
 \begin{equation}
 \lambda_{j,h}(\eta) = \frac{1}{h^{\nu}}\lambda_j(e^{-ih\eta}),
 \label{Definelambdajheta}
 \end{equation}
 where $\lambda_j(z)$ is the characteristic root of $\mathcal L(z,\overline{z}) - E_0I$. Their behavior at local extremal points plays an important role.  We consider only the minimum case, since the maximum case can be dealt with similarly.  We restrict ourselves to two cases: (1) simple isolated roots and (2) double roots.

\medskip
We assume for some $j$, $1 \leq j \leq {\mathsf s}$, there exists an $h$-independent open set $\mathcal K_0$  in ${\bf T}^d$ with the following properties. 

\medskip
\noindent
({\bf B-1}) {\it \ $\lambda_j(e^{-i\eta}) \geq 0$  on $\mathcal K_0$,  and there exists a unique  $d_1 \in \mathcal K_0$ such that $\lambda_{j}(e^{-id_1}) = 0$.}

\medskip
We assume one of the following conditions (B-2-1) or (B-2-2):

\medskip
\noindent
({\bf B-2-1}) \ \ {\it There exist  constants $\epsilon_1, \epsilon_2 > 0$ such that}
\begin{equation}
\lambda_{j-1}(e^{-i\eta}) < -\epsilon_1 < \lambda_j(e^{-i\eta}) < \epsilon_2 < \lambda_{j+1}(e^{-i\eta}), \quad \forall \eta \in \overline{\mathcal K_0}.
\nonumber
\end{equation}
\noindent
({\bf B-2-2}) \ \ {\it  $\lambda_{j-1}(e^{-i\eta}) = - \lambda_j(e^{-i\eta})$ on $\mathcal K_0$ and there exist constants $\epsilon_1, \epsilon_2 > 0$ such that
\begin{equation}
\lambda_{j-2}(e^{-i\eta}) <- \epsilon_1 <  \lambda_{j-1}(e^{-i\eta}) \leq \lambda_j(e^{-i\eta}) < \epsilon_2 <  \lambda_{j+1}(e^{-i\eta}), \quad \forall 
\eta \in \overline{\mathcal K_0}.
\nonumber
\end{equation}
}

\medskip
In both cases we put
\begin{equation}
\mathcal K = \mathcal K_0 - d_1 = \{\eta - d_1\, ; \, \eta \in \mathcal K_0\},
\nonumber
\end{equation}
\begin{equation}
P_h(\xi) = \lambda_{j,h}(\xi + d_h), \quad d_h = d_1/h,
\label{S3DefinePhxi}
\end{equation}
and assume

\smallskip
\noindent
({\bf B-3}) {\it For $0 \neq \xi \in \mathcal K/h$ the limit
\begin{equation}
P_h(\xi) \to P(\xi), 
\label{S3DefinePxi}
\end{equation}
together with that of all of its derivatives,
exists 
as $h \to 0$, where $P(\xi)$ is $C^{\infty}$ for $\xi \neq 0$, homogeneous of degree $\gamma > 0$ and 
\begin{equation}
CP(\xi) \leq P_h(\xi) \leq C^{-1}P(\xi), \quad 
{\it on} \quad \mathcal K/h
\nonumber
\end{equation}
for a constant $C > 0$.}

\medskip
Let $\Pi_h(\xi)$ be the eigenprojection associated with the eigenvalue $\lambda_{j,h}(\xi + d_h)$ for the case (B-2-1), and 
the sum of eigenprojections associated with $\lambda_{j-1,h}(\xi + d_h)$ and $\lambda_{j,h}(\xi + d_h)$ for the case (B-2-2). 
Since $\lambda_{j,h}$ (or the pair $\{\lambda_{j-1,h},  \lambda_{j,h}\}$) is isolated from the other eigenvalues, $\Pi_h(\xi)$ is smooth 
for $\xi \neq 0$. We assume:

\medskip
\noindent
({\bf B-4}) {\it For $0 \neq \xi \in \mathcal K/h$ there exists a projection $\Pi_0(\xi)$ such that 
\begin{equation}
\Pi_h(\xi) \to \Pi_0(\xi), 
\nonumber
\end{equation}
as $h \to 0$, together with all derivatives.
}

 In our applications, the resolvent of the matrix $\mathcal L_h(e^{- i(\xi + d_h)})$ converges together with its derivatives. Therefore, by using Risez' formula
$$
\frac{1}{2\pi i}\int_{C}\big(z - \mathcal L_h(e^{-i(\xi + d_h)}\big)^{-1}dz
$$
for the projection associated with (a group of) eigenvalues, where $C$ is the contour enclosing (the group of) eigenvalues in question, the assumption (B-4) is verified. See also \cite{Kato}, pp. 568-569.

\medskip
By (B-2-1) and (B-2-2) there exist constants $0 < C_0 < C_1$ and $0 < C_0' < C_1'$ such that 
\begin{equation}
\sup_{\xi \in \mathcal K}\lambda_{j,h}(\xi) \leq \frac{C_0}{h^{\nu}} < \frac{C_1}{h^{\nu}} < \inf_{\xi \in \mathcal K}
\lambda_{k,h}(\xi), \quad k \geq j+1,
\nonumber
\end{equation}
for both cases (B-2-1) and (B-2-2), and also 
\begin{equation}
\sup_{\xi \in \mathcal K}\lambda_{j,h}(\xi) \leq \frac{C'_0}{h^{\nu}} < \frac{C'_1}{h^{\nu}} < \inf_{\xi \in \mathcal K}
\big|\lambda_{k,h}(\xi)\big|, 
\nonumber
\end{equation}
for $k \leq j-1$ for the case (B-2-1), and $k \leq  j-2$ for the case (B-2-2).

Take an $h$-independent open interval $I  \subset\subset (0,\infty)$\footnote{By $A \subset\subset \mathcal O$ for an open set $\mathcal O$, we mean $\overline{A} \subset \mathcal O$ {and compact}.}, and $h_0$ small enough so that 
$I \subset\subset (0, C/h^{\nu})$, for $0 < h < h_0$, where $C = \min\{C_0,C_0'\}$. 
Then there exists a constant $C' > 0$ such that
\begin{equation}
\big|\lambda_{k,h}(\xi) - \lambda\big| \geq \frac{C'}{h^{\nu}}, \quad \forall \lambda \in I
\nonumber
\end{equation}
for $k \neq j$ in the case (B-2-1), and for $k \neq j, j-1$ in the case (B-2-2). This implies that
\begin{equation}
\big|\big(\mathcal L_h(e^{-ih(\xi +d_h)}) - z\big)^{-1}\big(1 - \Pi_h(\xi)\big)\big| \leq Ch^{\nu},
\label{S3Lh-z-11-Ph(xi)leqChnu}
\end{equation}
where $\big| \cdot \big|$ denotes the matrix norm, 
for ${\rm Re}\, z \in I$ and $\xi \in {\mathcal K}/h$. 


\begin{lemma}
\label{BlockdiagonalLemma}
For any $0 \neq \xi^{(0)} \in \mathcal K/h$, there exists an $h$-independent neighborhood $U \subset \big({\mathcal K}/h\big)\setminus\{0\}$ of $\xi^{(0)}$ and 
 a unitary matrix $A_h(\xi) \in C^{\infty}(U)$ such that 
\begin{equation}
A_h(\xi)^{\ast}\mathcal L_h(e^{-ih(\xi + d_h)})A_h(\xi) = 
\left(
\begin{array}{cc}
\lambda_j(e^{-ih(\xi + d_h)}) & 0 \\
0 & L_h(\xi)
\end{array}
\right), \quad \forall \xi \in U
\nonumber
\end{equation}
for the case (B-2-1), and 
\begin{equation}
A_h(\xi)^{\ast}\mathcal L_h(e^{-ih(\xi + d_h)})A_h(\xi) = 
\left(
\begin{array}{ccc}
\lambda_j(e^{-ih(\xi + d_h)}) & 0 & 0 \\
0 &- \lambda_j(e^{-i(\xi + d_h)}) & 0 \\
0 & 0 & L_h(\xi)
\end{array}
\right), \quad \forall \xi \in U
\nonumber
\end{equation}
for the case (B-2-2), where  $\mathcal L_h(\xi)$ is smooth with respect to $\xi \in U$ in both cases.
\end{lemma}

\begin{proof}
We consider the case (B-2-2). 
Take an orthonormal basis $v_1, \cdots, v_{\mathsf s}$ of ${\bf R}^{\mathsf s}$ such that $v_1, v_2 \in {\rm Ran}\,\Pi_0(\xi^{(0)})$, 
$v_3,\cdots,v_{\mathsf s} \in {\rm Ran}\, (1 - \Pi_0(\xi^{(0)}))$, and choose $w_i \in {\bf R}^{\mathsf s}$ such that 
$v_i = \Pi_0(\xi^{(0)})w_i$, $i = 1, 2$, $v_i = (1 - \Pi_0(\xi^{(0)}))v_i$, $i = 3, \cdots, {\mathsf s}$. Let $A_{hi}(\xi) = \Pi_h(\xi)w_i$, $i = 1, 2$, 
$A_{hi}(\xi) = (1 - \Pi_h(\xi))w_i$, $i = 3, \cdots, {\mathsf s}$. Letting $A_h(\xi)$ be the unitary matrix with column vectors $A_{hi}(\xi)$ and taking 
a sufficiently small neighborhood of $\xi^{(0)}$, we obtain the lemma.
\end{proof}

\subsection{Gauge transformation}
\label{subsec3.2GaugeTransform}
We define the gauge transformation $\mathcal G_h$ by
\begin{equation}
\big(\mathcal G_ha\big)(n) = e^{ihn\cdot d_h}a(n).
\nonumber
\end{equation}
Then we have
$$
{\mathcal F_{disc,h}}\mathcal G_h\mathcal L_h(S_h){\mathcal G_h}^{\ast}{\mathcal F_{disc,h}}^{-1}  = \mathcal L_h(e^{-ih(\xi + d_h)}).
$$
Multiplying by $\Pi_h(\xi)$ the problem is reduced to  the multiplication operator $\lambda_{j,h}(\xi+d_h) = P_h(\xi)$. 
  We show that, as $h \to 0$, it converges to the (pseudo-differential) operator $P(-i\partial_x)$ in an appropriate sense.


\section{The free equation}
\label{Section4Freeequation}


\subsection{Characteristic surfaces}
We use the notation introduced in the previous section. Recall that
$P_h(\xi) = P(\xi) + O(h)$.
For $0 \neq E \in I$ the hypersurfaces
\begin{equation}
M_{E,h} = \{\xi \in {\bf R}^d\, ; \, P_h(\xi) = E\}, \quad
M_{E} = \{\xi \in {\bf R}^d\, ; \, P(\xi) = E\},
\nonumber
\end{equation}
are compact. 
By (B-3)  there exist  constants $C_0, h_0 > 0$ such that for all $E \in I$
\begin{equation}
M_{E,h} \subset \{\xi \in {\bf R}^d\, ; \, C_0 < |\xi| < C_0^{-1}\}, \quad 0 < \forall h < h_0.
\label{MEhisC0xi}
\end{equation}
 Let $\epsilon_d > 0$ be such that 
$\{|\xi| < 2\epsilon_d \} \subset \mathcal K$. Take $\chi_d \in C_0^{\infty}({\bf R}^d)$ such that $\chi_d(\xi) = 1$ on $|\xi| < \epsilon_d$ and $\chi_{d}(\xi) = 0$ on $|\xi| > 2 \epsilon_d$.
Note that
\begin{equation}
\chi_d(h\xi) = \chi_d(h\xi)\chi_{{\bf T}^d}(h\xi),
\label{chid=chidchTd}
\end{equation}
where
\begin{equation}
\chi_{{\bf T}^d}(\xi) = \left\{
\begin{split}
&1, \quad {\rm if} \quad \xi \in [-\pi,\pi]^d, \\ 
& 0, \quad {\rm if} \quad \xi \not\in [-\pi,\pi]^d.
\end{split}
\right.
\nonumber
\end{equation}
Define $f_h'$ by
\begin{equation}
f_h' = \mathcal F_{disc,h}^{-1}\big(\chi_d(h\xi)\widehat f_h(\xi)\big),
\label{Definefhprime}
\end{equation}
where $f_h(n) = f(hn)$, and consider the equation
\begin{equation}
\big(- \Delta_{disc,h} - z\big)v_h' = \mathcal G_hf_h',
\label{S3Equationforvhprime}
\end{equation}
 where $\mathcal G_h$ is defined in Subsection \ref{subsec3.2GaugeTransform}.
Then $v_h = \mathcal G_h^{\ast}v_h'$ satisfies the gauge transformed equation
\begin{equation}
\mathcal G_h^{\ast}\big(\mathcal L_h(S_h) - z\big)\mathcal G_hv_h = f'_h.
\nonumber
\end{equation}
Assuming that $f \in L^2({\bf R}^d)\cap C({\bf R}^d)$ and $\big(f(hn)\big)_{n\in{\bf Z}^d} \in L^1({\bf Z}^d)$, we can solve it when $z \not\in {\bf R}$, i.e. we have
\begin{equation}
\widehat v_h(\xi,z) = \big(\mathcal L_h(e^{-ih(\xi + d_h)}) - z\big)^{-1}\chi_d(h\xi)\widehat f_h(\xi)  \in L^2({\bf T}^d_h),
\label{uh=fracfhPh-z}
\end{equation}
\begin{equation}
\widehat f_h(\xi) = \big(\frac{h}{2\pi}\big)^{d/2}\sum_{n\in{\bf Z}^d}f(hn)e^{-ihn\cdot \xi} \in L^2({\bf T}^d_h)\cap C({\bf T}^d_h).
\label{S4widehatfhxidefine}
\end{equation}
Note that $\big(f(hn)\big)_{n\in{\bf Z}^d} \in L^2({\bf Z}^d)$. 

By (\ref{S3Lh-z-11-Ph(xi)leqChnu}) and (\ref{AppendhdfhL2bfTdleqCfalphabetah<h0}), we have for $m > \big[d/2\big], s > d/2$,
\begin{equation}
\|\big(1 - \Pi_h(\xi)\big)\widehat v_h(\xi,z)\|_{L^2({\bf T}^d_h)} \leq Ch^{\nu}\|\widehat f_h\|_{L^2({\bf T}^d_h)} \leq Ch^{\nu/2}\|f\|_{m,s}.
\label{S4-PhxivhCh-nux}
\end{equation}

Therefore the part $\big(1- \Pi_h(\xi)\big)\widehat v_h(\xi, z)$ disappears as $h \to 0$. 
We consider the part $\Pi_h(\xi)\widehat v_h(\xi, z)$. 
Let us first consider the case (B-2-1). 
We put
\begin{equation}
\begin{split}
u_h(n,z) & = \big(\mathcal F_{disc,h}\big)^{-1}\chi_d(h\xi)\frac{\Pi_h(\xi){\widehat f_h(\xi)}}{P_h(\xi)-z} \\
& = \big(\frac{h}{2\pi}\big)^{d/2}\int_{{\bf R}^d}e^{ihn\cdot\xi}\chi_{{\bf T}^d}(h\xi)\chi_d(h\xi)\frac{\Pi_h(\xi){\widehat f_h(\xi)}}{P_h(\xi)-z} d\xi.
\end{split}
\nonumber
\end{equation}
A natural interpolation of $\big(u_h(n,z)\big)_{n\in{\bf Z}^d}$ on ${\bf R}^d$ is 
\begin{equation}
\begin{split}
\widetilde u_h(x,z) & = 
h^{d/2}\mathcal F_{cont}^{-1}\Big(\chi_{{\bf T}^d}(h\xi)\chi_d(h\xi)\frac{\Pi_h(\xi)\widehat f_h(\xi)}{P_h(\xi)-z}\Big) \\
& = \Big(\frac{h}{2\pi}\Big)^{d/2}\int_{{\bf R}^d}e^{ix\cdot\xi}\chi_{{\bf T}^d}(h\xi)\chi_d(h\xi)
\frac{\Pi_h(\xi)\widehat f_{h}(\xi)}{P_h(\xi)-z} d\xi. 
\end{split}
\label{hd/2F-1chi0hxiuhxiz}
\end{equation}
Take $\chi_1(\xi) \in C^{\infty}({\bf R}^d)$ such that
\begin{equation}
\chi_1(\xi)  = \left\{
\begin{split}
&1, \quad {\rm if} \quad |\xi| \leq C_0/3 \ \   {\rm or} \ \ 3/C_0 \leq |\xi|, \\
&0, \quad {\rm if} \quad 2C_0/3 \leq |\xi| \leq  2/C_0,
\end{split}
\right.
\label{Definechi1(xi)}
\end{equation}
where $C_0$ is from (\ref{MEhisC0xi}), 
and put
\begin{equation}
\chi_2(\xi) = 1 - \chi_1(\xi).
\nonumber
\end{equation}
Recalling that $P(\xi)$ is homogeneous of degree $\gamma$ by (B-3),  on the support of $\chi_1(\xi)$
\begin{equation}
\big|P_h(\xi) - z\big| \geq C(1 + |\xi|^{\gamma}), \quad {\rm for} \quad {\rm Re}\, z \in I \quad {\rm and} \quad \xi \in \mathcal K/h .
\label{Ph(xi)geqxi2}
\end{equation}
In view of (\ref{hd/2F-1chi0hxiuhxiz}) 
we put
\begin{equation}
\widetilde u_h^{(i)}(x,z) = h^{d/2}\mathcal F_{cont}^{-1}
\Big(\chi_{{\bf T}^d}(h\xi)\chi_i(\xi)\chi_d(h\xi)\frac{\Pi_h(\xi)\widehat f_h(\xi)}{P_h(\xi) - z}\Big) \in L^2({\bf R}^d), \quad i = 1, 2, 
\label{S4Definewidetildeuhi(x,z)}
\end{equation}
\begin{equation}
\widehat u^{(i)}_{h}(\xi,z) = \chi_{{\bf T}^d}(h\xi) \chi_i(\xi)\chi_d(h\xi)\frac{\Pi_h(\xi)\widehat f_h(\xi)}{P_h(\xi) - z} \in L^2({\bf T}^d_h), \quad 
i = 1, 2.
\label{Defineu(1)h(xi)}
\end{equation}


\subsection{Convergence outside $M_{E,h}$}

For $E\in I$ (\ref{Ph(xi)geqxi2}) implies 
\begin{equation}
\lim_{\epsilon\to0}\widehat u^{(1)}_{h}(\xi,E + i\epsilon) =
\widehat u^{(1)}_{h}(\xi,E + i0) = \chi_{{\bf T}^d}(h\xi)\chi_1(\xi)\chi_d(h\xi)\frac{\Pi_h(\xi)\widehat f_h(\xi)}{P_h(\xi) - E},
\end{equation}
since $P_h(\xi) - E \neq 0$. 
Let 
\begin{equation}
\widehat u^{(1)}(\xi,E+i0) =  \chi_1(\xi)\Pi_0(\xi)\frac{(\mathcal F_{cont}f)(\xi)}{P(\xi) - E} \in L^2({\bf R}^d),
\nonumber
\end{equation}
\begin{equation}
\widetilde u^{(1)}(x,E+i0) = \mathcal F_{cont}^{-1}\big(\widehat u^{(1)}(\xi,E+i0)\big) \in L^2({\bf R}^d).
\nonumber
\end{equation}


\begin{theorem}
\label{TheoremCovergenceoutsideMEh}
Assume that $f \in H^{m,s}({\bf R}^d)$ for some $m > \big[d/2\big]+1$, $s> d$. Then as $h\to0$
 $$
 \widetilde u^{(1)}_h(x,E + i0) \to \widetilde u^{(1)}(x, E + i0) \quad {\it in} \quad L^{2,-\delta}({\bf R}^d)
 $$
  for any $\delta > 0$.
\end{theorem}

\begin{proof}
We first show that for $h\to0$
\begin{equation}
 \widetilde u^{(1)}_h(x,E + i0) \to \widetilde u^{(1)}(x, E + i0) \quad {\it in} \quad \mathcal S'({\bf R}^d).
\label{S4ConvinSprime}
\end{equation}
In fact, by Lemma \ref{AppendPointwiseConv},
\begin{equation}
\widehat f_h(\xi)h^{d/2} = (2\pi)^{-d/2}\sum_{n\in{\bf Z}^d}f(hn)e^{-ihn\cdot\xi}h^d \to 
\big(\mathcal F_{cont}f\big)(\xi).
\nonumber
\end{equation}
The assumption (B-4) then yields (\ref{S4ConvinSprime}). 

Next we prove the convergence in $L^{2,-\delta}({\bf R}^d)$ by showing the following two facts:
\begin{enumerate}
\item
$\{\widetilde u^{(1)}_h(x, E + i0)\}_{0<h<h_0}$ is bounded in $H^{\gamma}({\bf R}^d)$. 
\item 
Compactness and uniqueness of accumulation point.
\end{enumerate}

We show the assertion (1). Using (\ref{Ph(xi)geqxi2})  and (\ref{S4Definewidetildeuhi(x,z)}), we have
\begin{equation}
\begin{split}
\|\widetilde u_h^{(1)}(x)\|^2_{H^{\gamma}({\bf R}^d)} & \leq Ch^d\|\widehat f_h\|^2_{L^2({\bf T}^d_h)} 
\leq C\|f\|_{H^{m,s}({\bf R}^d)}^2
\end{split}
\nonumber
\end{equation}
from (\ref{AppendhdfhL2bfTdleqCfalphabetah<h0}) in Lemma \ref{AppendPointwiseConv}. 
This proves (1).

Omitting $E + i0$, take any sequence $\widetilde u^{(1)}_{h_i}(x), \ i = 1, 2, \cdots$, $h_i \to 0$. 
For $\delta > 0$ the result (1) implies that
there exists a constant $C > 0$ independent of $h_i$ such that
$$
\int_{|x|>R}\langle x\rangle^{-\delta}|\widetilde u^{(1)}_{h_i}(x)|^2dx \leq CR^{- \delta}.
$$
In the bounded domain $\{|x| < R\}$,  applying Rellich's selection theorem,
one can choose a subsequence $\{\widetilde u^{(1)}_{h'_i}(x)\}$ of $\{\widetilde u^{(1)}_{h_i}(x)\}$ which converges to some $w$ in $L^{2,-s}({\bf R}^d)$. Note by Plancherel's formula, letting $(\, , \, )$ be the $L^2({\bf R}^d)$-inner product,
$$
(\widetilde u^{(1)}_{h'_i},\varphi) = (\widehat u^{(1)}_{h'_i},\mathcal F_{cont}\varphi), \quad \forall \mathcal F_{cont}\varphi \in C_0^{\infty}({\bf R}^d),
$$
Letting $h'_i \to 0$, we have $(w,\varphi) = (\widehat u^{(1)},\mathcal F_{cont}\varphi)$. This proves that the limit $w$ of the subsequence 
$\widetilde u^{(1)}_{h'_i}(x)$ does not depend on this subsequence, which proves the uniqueness of the accumulation point of $\{\widetilde u_h(x)\}$.  Then  $\widetilde u^{(1)}_{h}(x)$ itself converges to $\widetilde u^{(1)}(x)$ in $L^{2,-\delta}({\bf R}^d)$.
 \end{proof}


\subsection{Uniform estimates near $M_{E,h}$}
Let $\chi_1(\xi)$ be as in (\ref{Definechi1(xi)}), 
$\chi_2(\xi) = 1 - \chi_1(\xi)$, and 
$$
\Xi_2 = {\rm supp}\, \chi_2(\xi).
$$
There exists a constant $C > 0$ such that 
\begin{equation}
C|\xi|^{\gamma} \leq P_h(\xi) \leq C^{-1}|\xi|^{\gamma} \quad {\rm on} \quad \Xi_2, 
\nonumber
\end{equation}
moreover, as $h\to0$
$P_h(\xi) \to P(\xi)$. Since $P(\xi)$ is homogeneous, $\nabla_{\xi}P(\xi) \neq 0$ on $\Xi_2$. Therefore $\nabla P_h(\xi) \neq 0$ on $\Xi_2$.
Take $\xi^{(0)} = (\xi^{(0)}_1,\cdots,\xi^{(0)}_d) \in \Xi_2$ arbitrarily. Around $\xi^{(0)}$ 
we make a linear change of variables, $\xi \to \eta$, hence $\xi^{(0)} \to \eta^{(0)}$, 
so that in the $\eta$-coordinates $(1,0)$ is an outward transversal direction to $M_{E,h}$ at $\eta^{(0)}$, and the following factorization holds near $\eta^{(0)}$:
\begin{equation}
P_h(\xi(\eta)) - E - i\epsilon = \big(\eta_1 - p_h(\eta',E+ i\epsilon)\big)q_h(\eta,E+i\epsilon),
\nonumber
\end{equation}
\begin{equation}
{\rm Im}\, p_h(\eta',E + i\epsilon) \geq 0, \quad q_h(\eta,E+i\epsilon) \neq 0.
\nonumber
\end{equation}
For the simplicity of notation, we write $\xi$ instead of $\eta$.  Take small $\delta > 0$ so that 
$\xi_1 \neq 0$ for $|\xi - \xi^{(0)}| < \delta$.  Let 
$$
z = E + i\epsilon,
$$
and assume that ${\rm Re}\, z \in I$ and ${\rm Im}\, z > 0$. 
 Take $\chi(\xi) \in C_0^{\infty}({\bf R}^d)$ such that 
\begin{equation}
\chi(\xi) = 
\left\{
\begin{split}
& 0 \quad {\rm for} \quad |\xi - \xi^{(0)}| \geq 2\delta/3, \\
& 1 \quad {\rm on} \quad |\xi - \xi^{(0)}| \leq \delta/3,
\end{split}
\right.
\nonumber
\end{equation}
and put
\begin{equation}
\widehat g_h(\xi,z) = \frac{\chi(\xi)}{q_h(\xi,z)}\chi_{{\bf T}^d}(h\xi)\chi_2(\xi)\chi_d(h\xi)\Pi_h(\xi)\widehat f_h(\xi),
\label{defineghxiz}
\end{equation}
\begin{equation}
\widehat v_h(\xi,z) = \frac{\widehat g_h(\xi,z)}{\xi_1 - p_h(\xi',z)},
\label{DefineWidehatvh(xiz)=chqfh}
\end{equation}
\begin{equation}
\widetilde v_h(x,z) = h^{d/2}\mathcal F_{cont}^{-1}\big(\widehat v_h(\xi,z)\big).
\label{DefineWidetildevh(xiz)}
\end{equation}
Note that $\frac{\chi(\xi)}{q_h(\xi,z)}\chi_{{\bf T}^d}(h\xi)\chi_2(\xi)\chi_d(h\xi)\Pi_h(\xi) = \chi(\xi)\Pi_h(\xi)$, hence
$$
\Big|\partial_{\xi}^{\alpha}\frac{\chi(\xi)}{q_h(\xi,z)}\chi_{{\bf T}^d}(h\xi)\chi_2(\xi)\chi_d(h\xi)\Pi_h(\xi)\Big| \leq C_{\alpha}, \quad \forall \alpha
$$
for a constant $C_{\alpha}$ independent of $0 < h < h_0$ and $z$, ${\rm Re}\, z \in I$.


\begin{lemma} 
\label{Lemmavh(xi,z)limit}
The function $\widetilde v_h(x,z)$ defined by (\ref{DefineWidetildevh(xiz)}) has the following properties.

\noindent
(1) For $m > \big[d/2\big], s > d/2 + 1$, there exists a constant $C > 0$ such that
\begin{equation}
\|\widetilde v_h(x,z)\|_{{\mathcal B}^{\ast}({\bf R}^d)} \leq C\|f\|_{m,s}, \quad 0 < h < h_0.
\label{Uniforminepsilonbound}
\end{equation}
(2) The limit $\lim_{\epsilon\to0}\widehat v_h(\xi,E + i\epsilon)  = \widehat v_h(\xi,E + i0)$ exists in the weak-$\ast$ sense, i.e. 
\begin{equation}
(\widetilde v_h(x,E + i\epsilon), \varphi(x)) \to 
(\widetilde v_h(x,E + i0),\varphi(x)), \quad 
\forall \varphi(x) \in \mathcal B({\bf R}^d).
\nonumber
\end{equation}
(3) Moreover, $\widetilde v_h(x,E + i0)$ is an $L^2({\bf R}^{d-1})$-valued bounded function of $x_1$, and
\begin{equation}
\|\widetilde v_h(x_1,\cdot,E + i0)\|_{L^2({\bf R}^{d-1})} \to 0, \quad 
{\rm as} \quad x_1 \to - \infty.
\label{vhlimitx1to-infty}
\nonumber
\end{equation}
\end{lemma}

\begin{proof} 
(1) We only have to consider the case  $f \in C_0^{\infty}({\bf R}^d)$.
By definition,
\begin{equation}
\big(\xi_1 - p_h(\xi',z)\big)\widehat v_h(\xi,z) =  \widehat g_h(\xi,z).
\label{Eq:xi1-phxi,z)vh=ghg}
\end{equation}
We pass to the inverse Fourier transform with respect to $\xi_1$. Letting
$$
g_{h1}(y_1,\xi',z) = (2\pi)^{-1/2}\int_{-\infty}^{\infty}e^{iy_1\xi_1}\widehat g_h(\xi_1,\xi',z)d\xi_1,
$$
\begin{equation}
v_{h1}(x_1,\xi',z) = i\int_{-\infty}^{x_1}e^{i(x_1 - y_1)p_h(\xi',z)}g_{h1}(y_1,\xi',z)dy_1,
\label{vhlimitx1toplusinfty}
\end{equation}
we see that the Fourier transform with respect to $x_1$ of $v_{h1}(x_1,\xi',z)$ solves (\ref{Eq:xi1-phxi,z)vh=ghg}). 
Multiplying by $h^{d/2}$, and letting $\mathcal F_{x' \to \xi'}$ be the Fourier transform with respect to $x'$, we then 
have by (\ref{DefineWidetildevh(xiz)})
\begin{equation}
\widetilde v_h(x_1,x',z) = ih^{d/2}\int_{-\infty}^{x_1}\big(\mathcal F_{x' \to \xi'}\big)^{-1}
\left(e^{i(x_1 - y_1)p_h(\xi',z)}g_{h1}(y_1,\xi',z)\right)dy_1.
\label{vh(x)=inteix-ylambda+i0)gh(y)dy}
\end{equation}
Letting $\widetilde g_h(y_1,x',z) = h^{d/2}\big(\mathcal F_{x'\to\xi'}\big)^{-1}\big(g_{h1}(y_1,\xi',z)\big)$, we have
\begin{equation}
\|\widetilde v_h(x_1,\cdot,z)\|_{L^2({\bf R}^{d-1})} \leq \int_{-\infty}^{x_1}\|\widetilde g_h(y_1,\cdot,z)\|_{L^2({\bf R}^{d-1})}dy_1.
\label{widetildevhL2d-1isestimatedfromabovegh}
\end{equation}
Lemma \ref{Lemmafx1L2d-1leqfBnorm} then yields
\begin{equation}
\begin{split}
\|\widetilde v_h\|_{{\mathcal B}^{\ast}} & \leq \sqrt{2}\sup_{x_1\in{\bf R}}\|\widetilde v_h(x_1,\cdot,z)\|_{L^2({\bf R}^{d-1})} \\
& \leq \sqrt{2}\int_{-\infty}^{\infty}\|\widetilde g_h(y_1,\cdot,z)\|_{L^2({\bf R}^{d-1})}dy_1 \leq 2 \|\widetilde g_h\|_{{\mathcal B}({\bf R}^d)}.
\label{widetildvhestimaedfromaboveLemma4.2}
\end{split}
\end{equation}
 Note that
\begin{equation}
\|\widetilde g_h\|^2_{\mathcal B({\bf R}^d)} \leq 
C\|\widetilde g_h\|^2_{L^{2,1}({\bf R}^d)} \leq
Ch^d\|\widehat g_h(\xi)\|^2_{H^{1}({\bf R}^d)} 
\leq Ch^d\|\widehat f_h\|^2_{H^1({\bf T}^d_h)}.
\nonumber
\end{equation}
By Sobolev's inequality
$$
|f(x)| \leq C\|f\|_{H^{m,s}({\bf R}^d)}\langle x\rangle^{-s}, \quad 
m > \big[\frac{d}{2}\big], \quad s > 0.
$$
Then we have
\begin{equation}
\begin{split}
h^d\|\widehat f_h\|^2_{H^1({\bf T}^d_h)} &= h^d\sum_n|hn|^2|f(hn)|^2 \leq
C\|f\|^2_{H^{m,s}({\bf R}^d)}h^d\sum_n|hn|^2(1 + |hn|)^{-2s}\\
& \leq C\|f\|^2_{H^{m,s}({\bf R}^d)}\int_{{\bf R}^d}|x|^2(1 + |x|)^{-2s}dx.
\end{split}
\nonumber
\end{equation}
Therefore the right-hand side of (\ref{widetildvhestimaedfromaboveLemma4.2}) is dominated from above by $C\|f\|_{H^{m,s}({\bf R}^d)}$ if $s > d/2 +1$.
We have thus proven (\ref{Uniforminepsilonbound}).

The assertion (2) is a consequence of (\ref{vh(x)=inteix-ylambda+i0)gh(y)dy}), and we have
\begin{equation}
\widetilde v_h(x,E + i0) = i \int_{-\infty}^{x_1}
e^{i(x_1 - y_1)p_h(\xi',z)}\widetilde g_{h1}(y_1,x',E)dy_1.
\label{vhintegralepsiklon=0case}
\end{equation} 
The inequality (\ref{widetildevhL2d-1isestimatedfromabovegh}) implies (\ref{vhlimitx1to-infty}).
\end{proof}


\begin{definition}
For $u(\xi) \in S'({\bf R}^d)$, the wave front set $WF^{\ast}(u)$ is defined as follows. For $(\omega, \xi_0) \in S^{d-1}\times{\bf R}^d$, $(\omega,\xi_0) \not\in WF^{\ast}(u)$ if there exist $0 < \delta < 1$ and $\chi(\xi) \in C_0^{\infty}({\bf R}^d)$ 
such that $\chi(\xi_0)=1$ and
\begin{equation}
\lim_{R\to\infty}\frac{1}{R}\int_{|x|<R}\big|C_{\omega,\delta}(x)\mathcal F_{cont}^{-1}\big(\chi u\big)(x)|^2dx = 0,
\label{OutgoingradCond}
\end{equation}
where $C_{\omega,\delta}(x)$ is the characteristic function of the cone $\{x \in {\bf R}^d\, ; \, \omega\cdot x > \delta|x|\}$.
\end{definition}


\begin{definition}
\label{DefinitionPsiDORadCond}
Let $P_h(\xi)$ and $P(\xi)$ be as in (\ref{S3DefinePhxi}) and (\ref{S3DefinePxi}).
Let $\mathcal P^{(h)}_-$ be the set of $\Psi$DO's whose symbol $p_{-}(x,\xi)$ satisfies the following conditions:

\medskip
\noindent
(1)  $\big|\partial_{x}^{\alpha}\partial_{\xi}^{\beta}p_-(x,\xi)\big| \leq C_{\alpha\beta}\langle x\rangle^{-|\alpha|}$, \ $\forall \alpha, \beta$, 

\smallskip
\noindent
(2) There exist constants $0 < a < b<\infty$ and $-1 < \delta < 1$ such that
\begin{equation}
p_-(x,\xi) = 0 \quad {\rm for} \quad |\xi| \not\in (a,b),
\nonumber
\end{equation}
\begin{equation}
p_-(x,\xi) = 0 \quad {\rm for} \quad \frac{x}{|x|}\cdot\frac{\nabla_{\xi}P_h(\xi)}{|\nabla_{\xi}P_h(\xi)|} > \delta, \quad |x| > 1.
\label{Conditionp-=0ifxxi>delta}
\end{equation}
We define $\mathcal P_-$ in the same way with $P_h(\xi)$ replaced by $P(\xi)$.
\end{definition}

Let us now recall our notation. For $\widehat f_h(\xi)$ defined by (\ref{S4widehatfhxidefine}), we define
\begin{equation}
\widehat F_h(\xi) = \chi_{{\bf T}^d}(h\xi)\chi_d(h\xi)\Pi_h(\xi)\widehat f_h(\xi),
\nonumber
\end{equation}
\begin{equation}
\widehat U_h(\xi,z) = \widehat u_h^{(1)}(\xi,z) + \widehat u^{(2)}_h(\xi,z).
\label{S4SplitUhxiz}
\end{equation}
This is a solution to the equation
\begin{equation}
\big(P_h(\xi) - z\big)\widehat U_h(\xi,z) = \widehat F_h(\xi).
\nonumber
\end{equation}
Hence 
\begin{equation}
\widetilde U_h(x,z) = h^{d/2}\mathcal F_{cont}^{-1}\big(\widehat U_h(\xi,z)\big)
\label{S4DefineUhxz}
\end{equation}
converges as $z \to E + i0 \in I$.


\begin{lemma}
Assume that $m > \big[d/2\big], s > \max\{d/2 + 1, 5/2\}$. 
Let $\widehat U_h(\xi,E + i0)$ be as above. Then
\begin{equation}
WF^{\ast}(\widehat U_h) \subset \big\{\big(\omega_h(\xi), \xi\big)\, ; \, 
\xi \in M_{E,h}\big\},
\label{vhwavefrontset}
\end{equation}
where $\omega_h(\xi) = \nabla_{\xi}P_h(\xi)/|\nabla_{\xi}P_h(\xi)|$.
Moreover, for any $p_-(x,D_x) \in \mathcal P_-$ and  any $0 < \alpha < 1/2$, there exist constants $C > 0$, $h' > 0$ such that
\begin{equation}
\|p_-(x,D_x)\widetilde U_h\|_{L^{2,-\alpha}({\bf R}^d)} \leq C
\|f\|_{m,s}, \quad 0 < \forall h < h'.
\label{P-suhL2s-1leqfL2s}
\end{equation}
\end{lemma}

\begin{proof}
By (\ref{S4SplitUhxiz}), (\ref{S4DefineUhxz}), $\widetilde U_h(x,z)$ is split into 
$\widetilde U_h(x,z) = \widetilde u^{(1)}_h(x,z) + \widetilde u^{(2)}_h(x,z)$. In the proof of Theorem \ref{TheoremCovergenceoutsideMEh}, we showed that $\|\widetilde u_h^{(1)}(x,E)\|_{L^2({\bf R}^d)} \leq C\|f\|_{m,s}$. Therefore we prove the lemma for $\widetilde u_h^{(2)}(x,E+i0)$.

If $\xi \not\in M_{E,h}$ we have $\widehat U_h(\xi,E) = \widehat F_h(\xi)/(P_h(\xi) - E)$, which implies that
 for any $\omega \in S^{d-1}$, $(\omega,\xi) \not\in WF^{\ast}(\widehat u_h^{(2)})$ if $\xi \not\in M_{E,h}$. Take any $\xi^{(0)} \in M_{E,h}$ and $\epsilon > 0$, and $p_-(x,\xi) \in \mathcal P_-^{(h)}$ satisfying ${\rm supp}_{\xi} p_-(x,\xi) \subset 
 \{|\xi - \xi^{(0)}| < \epsilon\}$. By virtue of (\ref{Conditionp-=0ifxxi>delta}), for $t > 0$, we have on the support of $p_-(x,\xi)$ 
 \begin{equation}
 \begin{split}
 |\nabla_{\xi}\big(x\cdot\xi - tP_h(\xi)\big)|^2 &\geq |x|^2 - 2\delta t|x||\nabla_{\xi}P_h(\xi)| + t^2|\nabla_{\xi}P_h(\xi)|^2\\
 &\geq  (1 - \delta)\big(|x|^2 + t^2|\nabla_{\xi}P_h(\xi)|^2\big).
  \end{split}
  \nonumber
 \end{equation}
 Let $\widehat G_h (\xi) = \chi_2(\xi)\widehat F_h(\xi)$ and $\widetilde G_h = h^{d/2}\mathcal F_{cont}^{-1}\widehat G_h$.  Using the relation
$$
e^{i(x\cdot\xi - tP_h(\xi))} =-i \frac{\nabla_{\xi}\big(x\cdot\xi - tP_h(\xi)\big)}{|\nabla_{\xi}\big(x\cdot\xi - tP_h(\xi)\big)|^2}\cdot e^{i(x\cdot\xi - tP_h(\xi))},
$$
we have by integration by parts
\begin{equation}
\|\langle x\rangle^{-\alpha}p_-(x,D_x)e^{-itP_h(D_x)}\widetilde G_h\| \leq C\|q_h(t,x,D_x)e^{-itP_h(D_x)}\widetilde G_h\|,
\nonumber
\end{equation}
where $q_h(t,x,D_x)$ is a $\Psi$DO with symbol $q_h(t,x,\xi)$ satisfying
$$
|q_h(t,x,\xi)| \leq C\big(1 + t + |x|\big)^{-2}\langle x \rangle^{-\alpha}
$$
together with its derivatives. This implies
$$
\|p_-(x,D_x)e^{-itP_h(D_x)}\widetilde G_h\|_{-\alpha} \leq 
C(1 + t)^{-1-\epsilon_0}\|f\|_{m,s}
$$
for some $\epsilon_0 > 0$. Noting
$$
\int_0^{\infty}e^{-it(P_h(\xi) - E - i\epsilon)} = -i \big(P_h(\xi) - E - i\epsilon\big)^{-1}, \quad \epsilon > 0,
$$
and letting $\epsilon \to 0$, we obtain (\ref{P-suhL2s-1leqfL2s}) for $p_- \in \mathcal P_-^{(h)}$. This implies (\ref{vhwavefrontset}).
\end{proof}

We say that the solution of the equation $(P(D_x) - E)u = f$ satisfies the outgoing radiation condition if $p_-(x,D_x) u \in B^{\ast}_0({\bf R}^d)$ for any $p_- \in \mathcal P_-$. We call $u$  an outgoing solution.


\subsection{Convergence near $M_{E,h}$}
We have now constructed a solution to the discrete Schr{\"o}dinger equation
\begin{equation}
\big( - P_h(\xi) - E\big)\widehat U_h(\xi,E + i0) = \widehat F_h(\xi)
\nonumber
\end{equation}
having uniform estimates with respect to $0 < h < h_0$. Assume that
\begin{equation}
f \in H^{m,s}({\bf R}^d), \quad s > d + 1, \quad 
m > \frac{d}{2} + 1.
\nonumber
\end{equation}
Then as $h \to 0$ 
\begin{equation}
\widehat f_h(\xi) \to \widehat f(\xi) \quad {\rm pointwise}.
\nonumber
\end{equation}
In Theorem \ref{TheoremCovergenceoutsideMEh} the part outside of $M_{E,h}$ was shown to converge to the free Schr{\"o}dinger equation
$$
(- P(\xi) - E)U^{(1)} = \chi_1(\xi)\Pi_0(\xi)\widehat f(\xi),
$$
$ \widehat f(\xi)$ being the Fourier transform of $f$. 
We consider the part near $M_{E,h}$, i.e. $\widehat v_h(\xi, E + i0)$ constructed in the previous section.

In view of (\ref{vhlimitx1toplusinfty}) and the fact that 
$P_h(\xi) \to P(\xi)$, we then see that $\widetilde v_h(x)$ converges pointwise to some $\widetilde u(x)$, which satisfies the Schr{\"o}dinger equation
\begin{equation}
(P(D_x) - E)u = f_2,
\nonumber
\end{equation}
where $\widehat f_2(\xi) = \chi_2(\xi)\Pi_0(\xi)\widehat f(\xi)$. 
We have proven that $\widetilde v_h(x,E + i0)$ has the uniform estimates
\begin{equation}
\|\widetilde v_h\|_{\mathcal B^{\ast}({\bf R}^d)} < C,
\nonumber
\end{equation}
and satisfies the outgoing radiation condition uniformly in $0 < h < h_0$, i. e. (\ref{OutgoingradCond}) is satisfied uniformly in $0 < h < h_0$. Since 
$$
C^{-1}|\xi|^{\gamma} \leq P_h(\xi) \leq C|\xi|^{\gamma}
$$
holds, $\widetilde v_h$ is in $H^{\gamma}_{loc}({\bf R}^d)$ uniformly in $0 < h < h_0$. Then one can select a subsequence $h_j \to 0$ such that $\widetilde v_{h_j}$ converges to $v \in L^2_{loc}$ and $v$ satisfies 
\begin{itemize}
\item the Schr{\"o}dinger equation 
$
\big(P(D_x) - E\big)v = f_2,
$
\item 
$
v \in \mathcal B^{\ast}({\bf R}^d),
$
\item 
 the radiation condition. 
\end{itemize}
Here, we introduce a new assumption.

 \medskip
 \noindent
({\bf U-1}). {\it The solution of the equation
 \begin{equation}
 (P(D_x) - E)u = f \in \mathcal B
 \noindent
 \end{equation}
 satisfying $u \in \mathcal B^{\ast}$ and the radiation condition is unique. }
 
 \medskip
  Then $\{v_h\}_{0<h<h_0}$ converges as $h \to 0$. We have thus proven the following theorem.


\begin{theorem}
\label{TheoremwidetildeuhtowidetildeufordiscreteCase1}
Assume that $f \in H^{m,s}({\bf R}^d)$ for some $s > d + 1$ and $m > \big[d/2\big] + 1$. Assume (B-1), (B-2-1), (B-3), (B-4) and (U-1). 
Let $u_h(n,E + i0)$ be an outgoing solution to the gauge transformed equation
$$
(- \mathcal G_h^{\ast}\Delta_{disc,h}\mathcal G_h - E)u_h = f_h \quad {\it on} \quad {\bf Z}^d,
$$
where $f_h(n) = f(hn)$. 
 We put $\widehat u_h(\xi,E+i0) = \mathcal F_{disc,h}u_h$, and
\begin{equation}
\widehat v_h(\xi, E + i0) = \chi_d(h\xi)\Pi_h(\xi)\widehat u_h(\xi, E + i0).
\nonumber
\end{equation}
\begin{equation}
\widetilde v_h(x,E+i0) = \big(\frac{h}{2\pi}\big)^{d/2}\int_{{\bf T}^d_h}e^{ix\cdot\xi}\widehat u_h(\xi,E+i0)d\xi.
\nonumber
\end{equation}
Then the strong limit 
$$
\lim_{h\to0}\widetilde v_h(x,E+i\epsilon) = \widetilde v(x,E+i0) \quad \text{exists in} \quad L^{2, - 1/2 -\epsilon}({\bf R}^d), \quad \epsilon > 0,
$$
and $\widetilde v(x,E+i0)$ is the unique outgoing solution to the Schr{\"o}dinger equation 
$$
(P(D_x) - E)\widetilde v = g \quad {\it on} \quad {\bf R}^d,
$$
where $(\mathcal F_{cont} g)(\xi) = \Pi_0(\xi)(\mathcal F_{cont} f)(\xi)$, and $\widetilde v$ satisfies the radiation condition
$$
p_-(x,D_x)\widetilde v \in L^{2,- 1/2 + \epsilon}({\bf R}^d), \quad p_- \in \mathcal P_-.
$$
\end{theorem}


\subsection{Dirac equation}
\label{S4.5Dirac}
We next consider the case (B-2-2). Take $\xi^{(0)} \in \big(\mathcal K/h\big)\setminus\{0\}$ arbitraily. 
By Lemma \ref{BlockdiagonalLemma} there exists a neighborhood $\mathcal U$ of $\xi^{(0)}$ on which 
$\Pi_h(\xi)$ is into a sum 
$$
\Pi_h(\xi) = \Pi_h^{(+)}(\xi) + \Pi_h^{(-)}(\xi),
$$
where $\Pi_h^{(\pm)}(\xi)$ is the projection associated with the characteristic root $\pm P_h(\xi)$. Thus
\begin{equation}
\begin{split}
& \big(\mathcal L_h(e^{-ih(\xi + d_h)}) - z)^{-1} - z\big)^{-1}\Pi_h(\xi)  \\
& = 
\big(P_h(\xi)  - z\big)^{-1}\Pi_h^{(+)}(\xi) + \big(-P_h(\xi)  - z\big)^{-1}\Pi_h^{(-)}(\xi).
\end{split}
\nonumber
\end{equation}
The 1st term of the right-hand side is treated in the same way as in the previous subsection, and the 2nd term is 
easier to deal with since $- P_h(\xi) - E \neq 0$. We have thus obtained the following theorem.


\begin{theorem}
\label{TheoremwidetildeuhtowidetildeufordiscreteCase2}
Assume that $f \in H^{m,s}({\bf R}^d)$ for some $s > d + 1$ and $m > [d/2] + 1$. Assume (B-1), (B-2-2), (B-3), (B-4) and (U-1). 
Let $u_h(n,E + i0)$ be an outgoing solution to the gauge transformed equation
$$
(- \mathcal G_h^{\ast}\Delta_{disc,h}\mathcal G_h - E)u_h = f_h \quad {\it on} \quad {\bf Z}^d,
$$
where $f_h(n) = f(hn)$. 
 We put $\widehat u_h(\xi,E+i0) = \mathcal F_{disc,h}u_h$, and
\begin{equation}
\widehat v_h(\xi, E + i0) = \chi_d(h\xi)\Pi_h(\xi)\widehat u_h(\xi, E + i0).
\nonumber
\end{equation}
\begin{equation}
\widetilde v_h(x,E+i0) = \big(\frac{h}{2\pi}\big)^{d/2}\int_{{\bf T}^d_h}e^{ix\cdot\xi}\widehat u_h(\xi,E+i0)d\xi,
\nonumber
\end{equation}
\begin{equation}
(\mathcal F_{cont} g^{(+)})(\xi) = \Pi_0^{(+)}(\xi)(\mathcal F_{cont} f)(\xi).
\nonumber
\end{equation}
Then the strong limit
$$
\lim_{h\to0}\widetilde v_h(x,E+i\epsilon) = \widetilde v(x,E+i0) \quad \text{exists in} \quad L^{2, - 1/2 -\epsilon}({\bf R}^d), \quad \epsilon > 0.
$$
Here $\widetilde v(x,E+i0)$ is split into two parts
$$
\widetilde v(x,E+i0) = \widetilde v^{(+)}(x,E+i0) + \widetilde v^{(-)}(x,E+i0),
$$
 $\widetilde v^{(+)}(x,E+i0)$  being the unique solution to the Schr{\"o}dinger equation 
$$
(P(D_x) - E)\widetilde v^{(+)} = g^{(+)} \quad {\it on} \quad {\bf R}^d,
$$
 satisfying the outgoing radiation condition
$$
p_-(x,D_x)\widetilde v^{(+)} \in L^{2,- 1/2 + \epsilon}({\bf R}^d), \quad p_- \in \mathcal P_-,
$$
and $\widetilde v^{(-)}(x,E+i0)$  is the unique $L^2$-solution to the Schr{\"o}dinger equation 
$$
(-P(D_x) - E)\widetilde v^{(-)} = g^{(-)} \quad {\it on} \quad {\bf R}^d.
$$
\end{theorem}

In the application, we  encounter the case in which $d = 2$ and
\begin{equation}
P(\xi) = \Big(\sum_{i,j=1}^2a_{ij}\xi_i\xi_j\Big)^{1/2},
\nonumber
\end{equation}
where $\big(a_{ij}\big)$ is a positive definite symmetric matrix.
We can make a linear transformation $\xi \to \eta$ so that
$$
P(\xi) = \big(\eta_1^2 + \eta_2^2\big)^{1/2}.
$$
Let us note the following equality for $w \in {\bf C}$
\begin{equation}
\left(
\begin{array}{cc}
0& \overline{w} \\
w & 0
\end{array}
\right) = P^{\ast}\left(
\begin{array}{cc}
|w|& 0 \\
0 & - |w|
\end{array}
\right)P,
\nonumber
\end{equation}
\begin{equation}
P = \frac{1}{\sqrt2}\left(
\begin{array}{cc}
1& \frac{|w|}{w} \\
-1 & \frac{|w|}{w}
\end{array}
\right).
\nonumber
\end{equation}
Letting $w =  \eta_1 + i\eta_2$, we have
\begin{equation}
\left(
\begin{array}{cc}
0 & \overline{w} \\
w & 0
\end{array}
\right) = 
\left(
\begin{array}{cc}
0 & 1 \\
1 & 0
\end{array}
\right)\eta_1 +
\left(
\begin{array}{cc}
0 & -i \\
i & 0
\end{array}
\right)\eta_2.
\nonumber 
\end{equation}
Replacing $\eta_1$ by $\frac{1}{i}\frac{\partial}{\partial y_1}$, and $\eta_2$ by $\frac{1}{i}\frac{\partial}{\partial y_2}$, where
$x \to y$ is a linear transformation such that $x\cdot\xi = y\cdot\eta$, we have thus seen that the lattice Hamiltonian $- \Delta_{\Gamma_h}^{(0)}$ converges to the 2-dimensional massless Dirac operator.


\subsection{On the assumption (U-1)}
The assumption (U-1) holds true for the above type of $P(D_x)$. Assume that
$$
P(\xi) = \Big(\sum_{i,j=1}^da_{ij}\xi_i\xi_j\Big)^{1/2},
$$
where $\big(a_{ij}\big)$ is a positive definite symmetric matrix. 
It is well-known that if $V(x)$ is a real-valued function decaying to 0 at infinity, e.g. there exists $\epsilon>0$ such that $\partial_x^{\alpha}V(x) = O(|x|^{-|\alpha| - \epsilon_0})$ for all $\alpha$, then (U-1) holds for the equation 
$(P(D_x)^2 + V(x) -E)u = f$, where $E > 0$.

\begin{lemma}
Let $u \in \mathcal B^{\ast}({\bf R}^d)$ be a solution of the equation $(P(D_x) - E)u = 0$ for some $E > 0$. If $u$ satisfies the radiation condition, then $u = 0$.
\end{lemma}

\begin{proof}
Multiplying the equation by $P(D_x) + E$, we have $(P(D_x)^2 - E^2)u = 0$. Then the lemma follows.
\end{proof}


\section{Perturbation by a potential}
\label{SectionPotentialPerturbation}
Let us consider the case where the characteristic root $\lambda_{j,h}(\eta)$, defined in (\ref{S3DefinePhxi}), (\ref{Definelambdajheta}), has a unique global minimum. 
Assume that

\medskip
\noindent
({\bf C-1})  \ {\it $\lambda_j(e^{-i\eta}) \geq 0$, and there exists a unique 
$d_1 \in {\bf T}^d$ such that $\lambda_{j}(e^{-id_1}) = 0$. }

\medskip
\noindent
({\bf C-2}) \ $
\displaystyle{
\max_{\eta \in {\bf T}^d}\lambda_{j-1}(e^{-i\eta}) < 0 < 
\min_{\eta\in {\bf T}^d}\lambda_{j+1}(e^{-i\eta}).} $

\medskip
\noindent
({\bf C-3})\  {\it Letting $P_h(\xi) = \lambda_{j,h}(\xi + d_h)$ be as in (\ref{S3DefinePhxi}), we assume that $P_h(\xi) \to P(\xi)$ on ${\bf T}^d$, where
\begin{equation}
P(\xi) = \sum_{|\alpha| = 2m}a_{\alpha}\xi^{\alpha}, 
\nonumber
 \end{equation} 
$m$ being a positive integer. }

\medskip
\noindent
We add a scalar potential $V(x)$ to $- \Delta_{\Gamma_h}$. 
Assume that 

\medskip
\noindent
({\bf V-2}) {\it \   $V(x) \in H^s({\bf R}^d)$ with $s > d/2$ and is compactly supported. }

\medskip
\noindent
We finally assume the uniqueness of solutions to the Schr{\"o}dinger equation.

\medskip
 \noindent
({\bf U-2}) {\it The solution of the equation
 \begin{equation}
 (P(D_x) + V(x) - E)u = f \in \mathcal B
 \noindent
 \end{equation}
 satisfying $u \in \mathcal B^{\ast}$ and the radiation condition is unique. }
 
\medskip
We consider
\begin{equation}
\big(-\Delta_{disc,h} + V_{disc,h} - z\big)u_h = f_h,
\label{S4-PhxivhCh-nu}
\end{equation}
where
\begin{equation}
\big(V_{disc,h}u_h\big)(n) = V(hn) u_h(n).
\nonumber
\end{equation}
Take an open interval 
$$
I \subset\subset \big(0,\min_{\eta\in {\bf T}^d}\lambda_{j+1}(e^{-i\eta})\big),
$$
and assume that ${\rm Re}\, z \in I$. As was discussed in the previous section, we prove

\begin{itemize}
\item fixed $h$ limit (Lemma \ref{Lemmaepsilonyo0witpotential}), 
\item  uniform bound (Lemma \ref{Lemma5.2}), 
\item compactness (Lemma \ref{Lemma5.3}).
\end{itemize}

In \cite{AndIsoMor1} we have proven that there exists the unique solution to the equation (\ref{S4-PhxivhCh-nu}) satisfying the outgoing radiation condition.


\begin{lemma}
\label{Lemmaepsilonyo0witpotential}
Assume that 
$f \in H^{m,s}({\bf R}^d)$ for some $s > d + 1$ and $m > \big[d/2 \big] + 1$. Assume ${\rm Re}\, z \in I$, ${\rm Im}\, z > 0$, and let $u_h(z) = \{u_h(n,z)\}_{n \in {\bf Z}^d}$ be the unique $L^2$-solution to the equation (\ref{S4-PhxivhCh-nu}). 
Then 
$$
\lim_{\epsilon\to0} u_h(E+i\epsilon) = u_h(E+i0)\quad
\text{ in \ \ $L^{2,-s}(\mathbf{Z}^d_h)$, \ \ $s>\tfrac12$},
$$
and $u_h(E+i0)$ is the unique outgoing solution to (\ref{S4-PhxivhCh-nu}) with $z = E$. 
\end{lemma}

By (C-1) $\lambda_{j}(e^{-i\eta})$ attains a global minimum at $d_1$. Take 
$\chi_d(h\xi) = 1$ in the arguments of \S \ref{Section4Freeequation}. 
Let $u_h = u_h(E + i0)$ be as in Lemma \ref{Lemmaepsilonyo0witpotential}. Recall that 
\begin{equation}
\widehat u_h(\xi) = \mathcal F_{disc,h}u_h,
\quad
\widetilde u_h(x) = \big(\frac{h}{2\pi}\big)^{d/2}\int_{{\bf T}^d_h}e^{ix\cdot\xi}\widehat u_h(\xi,z)d\xi.
\nonumber
\end{equation}
Letting 
\begin{equation}
w_h = \mathcal G_h^{\ast}u_h, \quad 
g_h = \mathcal G_h^{\ast}f_h - V_{disc,h}w_h,
\nonumber
\end{equation}
we consider the gauge transformed equation
\begin{equation}
\big( - \mathcal G_h^{\ast}\Delta_{disc,h}\mathcal G_h - z\big) w_h = g_h.
\label{S5Gaugetransfeq}
\end{equation}
Then 
we have
\begin{equation}
{\widehat w}_h(E + i0) = \big(\mathcal L_h(e^{-ih(\xi + d_h)}) - E - i0\big)^{-1}\widehat g_h.
\nonumber
\end{equation}
Let $\Pi_0$ be the eigenprojection associated with $P_h(\xi)$. 
Arguing in the same way as in (\ref{S3Lh-z-11-Ph(xi)leqChnu}) and  using the assumption that ${\rm Re}\, z \in I$, we have the following inequality
\begin{equation}
\|\big(1 - \Pi_0\big)\widehat w_h(\xi,E+i0)\|_{L^2({\bf T}^d_h)} \leq Ch^{\nu}
\big(\|f\|_{m,s} + \|w_h\|_{-1-\epsilon}\big), 
\label{S4whleqChnu}
\end{equation}
for $\epsilon > 0$. 
Therefore we consider $v_h$, where
\begin{equation}
\widehat v_h(\xi,z) = \Pi_0\widehat w_h(\xi,z).
\nonumber
\end{equation}


\begin{lemma}
\label{Lemma5.2}
Let $u_h$ be as in Lemma \ref{Lemmaepsilonyo0witpotential}. 

\noindent
(1) For any $s$, $1/2 < s < 1$, there exists a constant $C_s > 0$ such that
\begin{equation}
\|\widetilde u_h\|_{L^{2,-s}({\bf R}^d)} \leq C \|\widetilde f_h\|_{L^{2,s}({\bf R}^d)}, \quad \text{for all \ $h$, \ \ $0<h<h_0$.}
\label{UniformestimatewidetildeuhRd}
\end{equation}
(2) Let $p_-(x,D_x) \in P_-$. Then for any $s$, $1/2 < s < 1$, there exists a constant $C_s > 0$ such that
\begin{equation}
\|p_-(x,D_x)\widetilde u_h\|_{L^{2,s-1}({\bf R}^d)} \leq C_s \|\widetilde f_h\|_{L^{2,s}({\bf R}^d)}, \ \ \text{for all \ $h$, \ \ $0<h<h_0$.}
\nonumber
\end{equation}
\end{lemma}

\begin{proof} 
We prove (1). If (\ref{UniformestimatewidetildeuhRd}) does not hold then for any $n \geq 1$ there exists $u_{h_n}, f_{h_n}$  satisfying $(- \Delta_{disc,h_n} + V_{h_n} - E)u_{h_n} = f_{h_n}$ and $\|\widetilde u_{h_n}\|_{L^{2,-s}({\bf R}^d)} \geq n\|\widetilde f_{h_n}\|_{L^{2,s}({\bf R}^d)}$. Dividing by $\|\widetilde u_{h_n}\|_{L^{2,-s}({\bf R}^d)}$, there exist an outgoing $u_{h_n} \in \mathcal B^{\ast}({\bf Z}^d_{h_n})$ and $f_{h_n} \in \mathcal B({\bf Z}^d_{h_n})$ such that 
\begin{equation}
(- \Delta_{disc,h_n} + V_{h_n} - E)u_{h_n} = f_{h_n},
\nonumber
\end{equation}
\begin{equation}
\|\widetilde u_{h_n}\|_{L^{2,-s}({\bf R}^d)} = 1, \quad \|\widetilde f_{h_n}\|_{L^{2,s}({\bf R}^d)} \to 0, \quad n \to \infty.
\nonumber
\end{equation}
We put
$$
w_{h_n} = \mathcal G_{h_n}^{\ast}u_{h_n}, \quad 
{\widehat v}_{h_n}(\xi) = \Pi_0\widehat w_{h_n}(\xi),
$$
$$
\widehat v'_{h_n}(\xi) = \big(1 - \Pi_0\big)\widehat w_{h_n}(\xi).
$$
By (\ref{S4whleqChnu}), we have
\begin{equation}
\widehat v'_{h_n} \to 0, \quad {\rm in} \quad  L^{2,-s}  \quad {\rm for} \quad  s > 1/2.
\label{S5vhnprimeto0}
\end{equation}
By the a priori estimate (\ref{AprioriestimateinL2-s}), one can select a subsequence of $\{v_{h_n}\}$, which is denoted by $\{v_{h_{i}}\}$, such that $\{\widetilde v_{h_{i}}\}$ converges in $L^2_{loc}({\bf R}^d)$. Let $\widetilde u_{h_{i}}(x) \to v(x)$. One can show that 
\begin{equation}
v_{h_i}(x) \to v(x) \quad {\rm in} \quad L^{2,-s}({\bf R}^d).
\label{widetildeuhitow(x)}
\end{equation}
In fact, take $1/2 < t < s$. The resolvent estimate for $- \Delta_{disc,h}$ also holds between $L^{2,t}$ and $L^{2,-t}$. Therefore,
$$
\|\widetilde v_{h_i}\|_{L^{2,-t}({\bf R}^d)} < C,
$$
for a constant $C$ independent of $h_i$. This yields
$$
\int_{|x| > R}(1 + |x|^2)^{-s}|\widetilde v_{h_i}(x)|^2dx \to 0, \quad (R \to \infty)
$$
uniformly with respect to $h_i$. Thus we can  choose $h_i$ so that (\ref{widetildeuhitow(x)}) holds. In particular, we have
\begin{equation}
\|v\|_{L^{2,-s}({\bf R}^d)} = 1.
\nonumber
\end{equation}
Recall that
$$
\widetilde v_{h}(hn) = v_{h}(n). 
$$
Then by the definition of the Riemann integral we have
\begin{equation}
h_i^{d}\sum_{n\in {\bf Z}^d}V_{h_i}(n)v_{h_i}(n)\varphi(h_in) \to \int_{{\bf R}^d}V(x)\widetilde v(x)\varphi(x)dx, 
\quad \forall \varphi \in C_0^{\infty}({\bf R}^d).
\nonumber
\end{equation}

We can assume that $h_i \to h_{\infty} \geq 0$. 
If $h_{\infty} > 0$, we easily arrive at the contradiction, since we have the resolvent estimate for $- \Delta_{disc,h_{\infty}}$.
We  consider the case $h_{\infty} = 0$. In the equation
\begin{equation}
\big(- \Delta_{disc,h_i} + V_{h_i} - E\big)u_{h_i}(n) = f_{h_i}(n),
\nonumber
\end{equation}
we split $u_{h_n}$ into $u_{h_n} = \mathcal G_{h_n}w_{h_n} = 
\mathcal G_{h_n}(v_{h_n} + v'_{h_n})$. Recalling (\ref{S5vhnprimeto0}),  we take the inner product with $\varphi(h_i n)$ and sum up in $n$. Letting $h_i \to 0$, we have
\begin{equation}
(P(D_x) + V - E)v = 0.
\nonumber
\end{equation}
By the well-known results on resolvent estimates for Schr{\"o}dinger equations, we have
\begin{equation}
v \in {\mathcal B}^{\ast}({\bf R}^d), \quad p(x,D_x)v \in L^{2,s-1}({\bf R}^d),
\nonumber
\end{equation}
for $p(x,D_x) \in P_-$ and  $1/2 < s < 1$. Therefore $v$ is a solution to the homogeneous Schr{\"o}dinger equation satisfying the outgoing radiation condition, hence vanishes identically. 
This is a contradiction. We have thus proven (1). The assertion (2) follows from the resolvent estimate for $-\Delta_{disc,h}$. 
\end{proof}


\begin{lemma} 
\label{Lemma5.3}
Let $\widetilde u_h(x,E + i0)$ be as in Lemma \ref{Lemmaepsilonyo0witpotential}.\\
\noindent
(1) The set $\{\widetilde u_h(x,E + i0)\, ; \, 0 < h < h_0\}$ is compact in $L^2_{loc}({\bf R}^d)$. \\
\noindent
(2) The set $\{V_h\widetilde u_h(x,E + i0)\, ; \, 0 < h < h_0\}$ is compact in $L^{2,s}({\bf R}^d)$ for any $s > 0$, where
$V_h$ is defined by
$$
\big(V_hw\big)(x) = V(hn)w(x), \quad {\it if} \quad
x\in hn+[-h/2,h/2]^d.
$$
\end{lemma}

\begin{proof}
The assertion (1) follows from Lemma \ref{LemmaEstimatesofwidetildeuhinL2-s}. Since $V(x)$ is compactly supported, the assertion (2) follows from (1).
\end{proof}

These preparations and the assumption (U-2) are sufficient  to show the following theorem.


\begin{theorem}
\label{MainTheoremSchroedinger}
Let $s >1/2$.
As $h \to 0$, $\widetilde u_h(x) \to u(x)$, where $u(x) \in L^{2,-s}({\bf R}^d)$ and satisfies
\begin{equation}
u = \Pi_0u, \quad 
(P(D_x) + V(x) - E)u = \Pi_0f,
\nonumber
\end{equation}
\begin{equation}
p_-(x,D_x)u \in L^{2,s-1}({\bf R}^d), \quad \forall p_-(x, D_x) \in P_-.
\nonumber
\end{equation}
\end{theorem}

\begin{cor}
With the same notation as in Theorem \ref{MainTheoremSchroedinger}, 
$\widetilde u_h(x) \to u(x)$ locally uniformly on ${\bf R}^d$.
\end{cor}

\begin{proof}
We have for $m > d/2$ and $s > 1/2$,
\begin{equation}
\sup_{0<h<h_0}\|\widetilde u_h\|_{H^{m,-s}} < \infty.
\nonumber
\end{equation}
Using the a priori estimate and the Sobolev inequality, we get the assertion.
\end{proof}

 
\section{Complex energy}
In the above arguments the radiation condition was used 
at the step of the uniqueness of solutions to the equation
$ (P(D_x) + V - E)u = f$ for $E >0$. 
Usually this fact is proven by the Rellich type theorem and the unique continuation property for the Helmholtz equation. However, on some lattices, the latter result does not hold. 
Even for this case, if $E \not\in {\bf R}$, one can employ a simpler condition for the uniqueness. For the sake of simplicity we consider here the case $P(D_x) = - \Delta_{cont}$.


\begin{lemma}
\label{L2-suniqueness}
Assume that $u \in L^{2,-s}({\bf R}^d)\cap H^2_{\rm loc}(\mathbf{R}^d)$ satisfies $ (- \Delta_{cont} + V - z)u = 0$ on ${\bf R}^d$. If $z \not\in {\bf R}$ and $0 < s \leq 1/2$, then $u = 0$. 
\end{lemma}

\begin{proof}
Since $u, \frac{\partial u}{\partial r} \in L^{2,-s}$, we have
$\liminf_{r\to\infty} r^{n-2s}\int_{|x|=r}\Big|\overline{u}\frac{\partial u}{\partial r}\Big|dS = 0$. 
By integration by parts, we have
\begin{equation}
\int_{|x|<R}\Big((\Delta u) \overline u - u\overline{\Delta u}\Big)dx  = \int_{|x|=R}\Big(\frac{\partial u}{\partial r}\overline u - u\overline{\frac{\partial u}{\partial r}}\Big)dS.
\nonumber
\end{equation}
We then have, by taking the imaginary part and letting $R \to \infty$ along a suitable sequence, 
${\rm Im}\, z\int_{{\bf R}^d}|u|^2 dx = 0$, which proves the lemma.
\end{proof}

Arguing as in the previous sections, 
noting that $P_h(\xi) - z \neq 0$ for $z \not\in {\bf R}$, one can derive the estimates of $\widetilde u_h(x,z)$ in $H^2({\bf R}^d)$ uniformly with respect to $0 < h < h_0$.
Then by virtue of  Lemma \ref{L2-suniqueness}, for any $0 < s < 1/2$, one can conclude the convergence of $\widetilde u_h(x,z)$ in $L^{2,-s}({\bf R}^d)$ as $h \to 0$.  In the following Theorems \ref{TheoremwidetildediscretecomplexzB} and \ref{TheoremwidetildediscretecomplexzC}, we do not assume the unique continuation property for $\Delta_{disc,h}$.


\begin{theorem}
\label{TheoremwidetildediscretecomplexzB}
Assume that $f \in H^{m,s}({\bf R}^d)$ for some $s > d + 1$ and $m > [d/2] + 1$. Assume (B-1), (B-2-1) or (B-2-2), (B-3), (B-4) and (U-1). 
Let $z \not\in {\bf R}$, and $u_h(n,z)$ be an $L^2$-solution to the gauge transformed equation
$$
(- \mathcal G_h^{\ast}\Delta_{disc,h}\mathcal G_h - E)u_h = f_h \quad {\it on} \quad {\bf Z}^d,
$$
where $f_h(n) = f(hn)$. 
Then the strong  limit
$$
\lim_{h\to0}\widetilde v_h(x,z) = \widetilde v(x,z) \quad \text{exists in} \quad L^{2, - s}({\bf R}^d), \quad 0 < s < 1/2.
$$
This convergence is locally uniform on ${\bf R}^d$.

 For the case (B-2-1),  $\widetilde v(x,E+i0)$ is the unique $L^2$-solution to the equation
$$
(P(D_x) - z)\widetilde v = g \quad {\it on} \quad {\bf R}^d,
$$
 where $g$ is defined in (\ref{defineghxiz}). 

For the case (B-2-2), $v(x,z)$ split into two parts
$$
\widetilde v(x,z) = \widetilde v^{(+)}(x,z) + \widetilde v^{(-)}(x,z),
$$
 $\widetilde v^{(\pm)}(x,z)$  being the unique solution to the Schr{\"o}dinger equation 
$$
(P(D_x) - z)\widetilde v^{(\pm)} = g^{(\pm)} \quad {\it on} \quad {\bf R}^d.
$$
\end{theorem}


\begin{theorem}
\label{TheoremwidetildediscretecomplexzC}
Assume that $f \in H^{m,s}({\bf R}^d)$ for some $s > d + 1$ and $m > [d/2] + 1$. In addition to (B-1), (B-2-1) or (B-2-2), (B-3), (B-4) and (U-1), assume (C-1), (C-2), (C-3), and (U-2). Let $z \not\in {\bf R}$, and $u_h(n,z)$ be an $L^2$-solution of the equation
$$
(- \mathcal G_h^{\ast}\Delta_{disc,h}\mathcal G_h  + V_{disc,h}- z)u_h = f_h \quad {\it on} \quad {\bf Z}^d,
$$
where $f_h(n) = f(hn)$. 
Then the strong  limit
$$
\lim_{h\to0}\widetilde u_h(x,z) = \widetilde u(x,z) \quad \text{exists in} \quad L^{2, - s}({\bf R}^d), \quad 
0 < s < 1/2.
$$
The convergence is locally uniform on ${\bf R}^d$. Moreover, $u$ satisfies
\begin{equation}
u = \Pi_0u, \quad 
\big(P(D_x) + V(x) - z)u = \Pi_0f.
\nonumber
\end{equation}
\end{theorem}


\section{Derivation of Schr{\"o}dinger equations}

In the remaining sections, we apply the results in the previous sections to Hamiltonians on periodic lattices appearing in material science. 
We pick up basic examples whose spectral properties are studied in \cite{AndIsoMor1}.


\subsection{Square lattice}
Define the Laplacian on the square lattice in ${\bf R}^d$  by
\begin{equation}
\Delta_{\Gamma} = \sum_{j=1}^d\big(S_{j} + S_{j}^{\ast}\big).
\nonumber
\end{equation} 
In this case, we take the reference energy to be $E_0 = - 2d$ and consider
$$
\frac{1}{h^2}\sum_{j=1}^d\big(2 - S_{h,j} - S_{h,j}^{\ast}\big) + V_{disc,h}.
 $$
The symbol of $\dfrac{1}{h^2}\sum_{j=1}^d\big(2 - S_{h,j} - S_{h,j}^{\ast}\big)$ is 
\begin{equation}
\mathcal L_h(e^{-ih\eta}) = P_h(\eta) = \frac{1}{h^2}\sum_{j=1}^d\big(2 - e^{ih\eta_j} - e^{-ih\eta_j}\big) = \frac{4}{h^2}\sum_{j=1}^d\sin^2\frac{h\eta_j}{2},
\nonumber
\end{equation}
which satisfies the assumptions in \S 4 and \S 5 with $\mathcal G_h = 1$. Therefore by Theorem \ref{MainTheoremSchroedinger}, we obtain the following theorem.


\begin{theorem} 
\label{S7Schroedingerlimit}
Assume that $V(x) \in H^s({\bf R}^d)$ with $s > d/2$ and compactly supported,  and that $f \in H^{m,s}({\bf R}^d)$ with $m > d/2 + 1$, $s > d+1$. 
Let $E>0$ and take $h_0>0$ such that $E<4/h_0^2$.
Then the solution of the discrete 
Schr{\"o}dinger equation $(\mathcal L_h(S_h) + V_{disc,h} - E)u_h = f_h$, $0<h<h_0$, satisfying the radiation condition converges to that of 
$$
( - \Delta + V(x) - E)u = f \quad {\it on} \quad {\bf R}^d
$$
as $h\to0$.
\end{theorem}

Since this solution is written as $u = (- \Delta + V(x) - E - i0)^{-1}f$, we can use it to represent the $S$-matrix.


\subsection{Convergence of the S-matrix}
We recall the definition
\begin{equation*}
M_{E,h}=\{\xi\in\mathbf{T}^d_h\,|\,P_h(\xi)=E\}.
\end{equation*}
This surface is not smooth for $E=4j/h^2$, $j=1,2,\ldots,d-1$. Therefore we impose the condition $0<E<4/h^2$ in the sequel. In particular, we consider the $S$-matrix only in the energy interval $(0,4/h^2)$.

Fix $E\in(0,4/h^2)$. We recall that for such $E$ the surface $M_{E,h}$ is diffeomorphic to $S^{d-1}$. We use the parametrization
\begin{equation}\label{Generalizedeigenxij}
\xi_j=\frac{2}{h}\arcsin(\tfrac12 h\sqrt{E}\omega_j), \quad j=1,2,\ldots,d,\quad \omega\in S^{d-1}.
\end{equation}
We also recall that $\xi_j=\sqrt{E}\omega_j+O(h^2)$ as $h\to0$.

Let $dS_{E,h}$ denote the surface measure on $M_{E,h}$ induced by $d\xi$. Then
\begin{equation}
d\xi=\frac{1}{|\nabla_{\xi}P_h(\xi)|}dS_{E,h}dE.
\end{equation}
Let $L^2(M_{E,h})$ be the Hilbert space equipped with the inner product
\begin{equation*}
(\phi,\psi)_{L^2(M_{E,h})} = \int_{M_{E,h}}\phi(\xi)\overline{\psi(\xi)}dS_E.
\end{equation*}
Let
\begin{equation}
\Omega_h=\bigcup_{0<E<4/h^2}M_{E,h}.
\end{equation}
The characteristic function of this set is denoted by $\chi_{\Omega_h}$.

Let $f\in\mathcal{S}(\mathbf{R}^d)$. Put
\begin{equation}
(\mathcal{F}_{0h}(E)f)(\xi)=
\chi_{\Omega_h}(\xi)|\nabla_{\xi}P_h(\xi)|^{-1/2}
\big(\frac{h}{2\pi}\big)^{d/2}\sum_{n\in\mathbf{Z}^d}
e^{-ihn\cdot \xi}f(n),
\end{equation}
and then define
\begin{equation}
(\mathcal{F}_{0h}f)(E,\xi)=(\mathcal{F}_{0h}(E)f)(\xi).
\end{equation}
Let $E_{0h}$ denote the spectral measure of the operator 
$H_{0h} = - \Delta_{disc,h}$. Then 
\begin{equation}
\mathcal{F}_{0h}\colon E_{0h}((0,4/h^2))L^2(\mathbf{Z}^d_h)
\to L^2((0,4/h^2);L^2(M_{E,h});dE)
\end{equation}
is unitary. Note  that with the above definitions we have $\mathcal{F}_{0h}=\mathcal{F}_{0h}E_{0h}((0,4/h^2))$.

Define $H_{0h} = - \Delta_{disc,h}$ and $H_h = - \Delta_{disc,h} + V_{disc,h}$. Then the wave operators
\begin{equation}
W_h^{(\pm)} = \operatorname*{\mathrm s-lim}_{t\to\pm\infty}e^{itH_h}e^{-itH_{0h}}E_{0h}((0,4/h^2))
\nonumber
\end{equation}
exist and are asymptotically complete in the localized sense. Define the localized scattering operator by
\begin{equation}
S_h = \big(W_h^{(+)}\big)^{\ast}W_h^{(-)}.
\nonumber
\end{equation}
Then its localized Fourier transform 
$$
\widehat S_h := \mathcal F_{0h}S_h\big(\mathcal F_{0h}\big)^{\ast}
$$
has the direct integral representation
\begin{equation}
\widehat S_h = \int_0^{4/h^2}\widehat S_h(E)dE,
\nonumber
\end{equation}
where $\widehat S_h(E)$ is Heisenberg's $S$-matrix for $E\in(0,4/h^2)$, which is unitary on $L^2(M_{E,h})$. 
The scattering amplitude $A_h(E)$ is then defined by
\begin{equation}
\widehat S_h(E ) = I - 2\pi i A_h(E).
\nonumber
\end{equation}
It has the representation
\begin{equation*}
\begin{split}
A_h(E) = & \mathcal F_{0h}(E)V_h\mathcal F_{0h}(E)^{\ast} -
 \mathcal F_{0h}(E)V_hR_h(E + i0)V_h\mathcal F_{0h}(E)^{\ast}. 
\end{split}
\end{equation*}
Therefore for $E\in(0,4/h^2)$ its integral kernel is
\begin{equation*}
\begin{split}
A_h(E;\xi,\eta) &=  a_h(\xi,\eta)\big(\frac{h}{2\pi}\big)^d\sum_{n \in {\bf Z}^d}V(hn) e^{-ihn\cdot(\xi - \eta)} \\
& - a_h(\xi,\eta)\big(\frac{h}{2\pi}\big)^d\sum_{n\in{\bf Z}^d}
V(hn)e^{ihn\cdot\eta}u_h(E,\xi,n).
\end{split}
\end{equation*}
where
\begin{equation*}
u_h(E,\xi) = R_h(E + i0)\psi_h(\xi), 
\end{equation*}
$\psi_h(\xi) \in L^2({\bf Z}^d_h)$ is defined by
\begin{equation*}
\psi_h(\xi,n) = V(hn)e^{ihn\cdot\xi},
\end{equation*}
and
\begin{equation*}
a_h(\xi,\eta) = |\nabla P_h(\xi)|^{-1/2} |\nabla P_h(\eta)|^{-1/2}.
\end{equation*}
Using the parametrization $\xi = \xi_h(\omega)$ in (\ref{Generalizedeigenxij}), we can regard $A_h(E,\xi,\eta)$ as a function on $S^{d-1}\times S^{d-1}$, which we denote $A_h(E;\omega,\omega')$ for the sake of simplicity.

For the continuous case $A(E;\omega,\omega')$ is written as 
\begin{equation}
\begin{split}
A(E;\omega,\omega') &= C(E)\int_{{\bf R}^d}e^{-i\sqrt{E}(\omega - \omega')\cdot x}V(x)dx \\
& -
C(E)\int_{{\bf R}^d}e^{-i\sqrt{E}\omega\cdot x}V(x)u(E,x,\omega')dx,
\end{split}
\nonumber
\end{equation}
\begin{equation}
u(E;x,\omega') = R(E + i0)\big(V\psi(E,\omega')\big),
\nonumber
\end{equation}
\begin{equation*}
\psi(E,\omega;x) = V(x)e^{i\sqrt{E}\omega'\cdot x}
\end{equation*}


\begin{theorem}
\label{ConvergnceofS-matrix}
We have as $h \to 0$
\begin{equation*}
A_h(E;\omega,\omega') \to A(E;\omega,\omega').
\end{equation*}
\end{theorem}

\begin{proof}
By Lemma \ref{AppendPointwiseConv}
\begin{equation*}
h^d\sum_{n\in{\bf Z}^d}V(hn) e^{ihn\cdot(\xi-\eta)} \to
\int_{{\bf R}^d}V(x) e^{-ix\cdot(\xi-\eta)}dx,
\end{equation*}
\begin{equation*}
h^d\sum_{n\in{\bf Z}^d}V(hn)e^{-ihn\cdot\eta}u_h(n) \to \int_{{\bf R}^d}V(x)e^{-ix\cdot\eta}\widetilde u(x)dx,
\end{equation*}
where
$$
u_h(n) = u_h(E,\xi,n), \quad \widetilde u(x) = \widetilde u(E,x,\omega').
$$
Recall that $\widetilde u_h(x)$ converges to $\widetilde u(x)$ locally uniformly on ${\bf R}^d$, and $\widetilde u_h(hn) = u_h(n)$.
\end{proof}


\subsection{Triangular lattice}
\begin{figure}[hbtp]
\centering
\includegraphics[width=8cm,height=6cm,keepaspectratio]{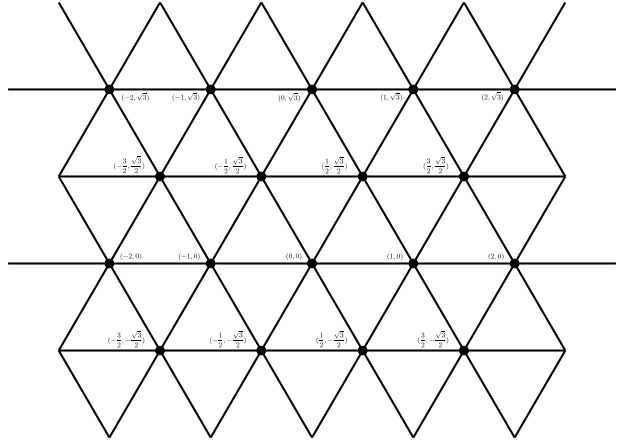}
\caption{Triangular lattice}
\label{Fig:TriangularLattice}
\end{figure}

The Laplacian for the triangular lattice is 
\begin{equation}
 \Delta_{\Gamma_h} = 
 \frac{1}{6h^2}
\big(S_{h,1} + S_{h,1}^{\ast} + S_{h,2} + S_{h,2}^{\ast} 
+ S_{h,1}S_{h,2}^{\ast} + S_{h,1}^{\ast}S_{h,2}\big).
\nonumber
\end{equation}
The reference energy is $E_0 = - 1/h^2$, and
the symbol of $- \Delta_{\Gamma_h}$  is given by
\begin{equation}
\begin{split}
P_h(\xi) & = 
\frac{1}{3h^2}\big(3 - \cos h\xi_1 - \cos h\xi_2 - \cos(h\xi_1-h\xi_2)\big) \\
& = \frac{2}{3h^2}\Big(\sin^2\frac{h\xi_1}{2} + \sin^2\frac{h\xi_2}{2} + \sin^2\frac{h(\xi_1-\xi_2)}{2}\Big).
\end{split}
\nonumber
\end{equation}
It has the unique global minimum at $\xi = 0$, and we have the asymptotic expansion
\begin{equation}
 P_h(\xi)  = \frac{1}{3}\big(\xi_1^2 - \xi_1\xi_2 + \xi_2^2\big) + O(h^2).
 \nonumber
\end{equation}
Therefore by the same argument as above, the following theorem holds.


\begin{theorem}
Assume that $V(x) \in H^s({\bf R}^2)$ with $s > 1$ and compactly supported, and that $f \in H^{m,s}({\bf R}^2)$ with $m > 2$, $s > 3$. Then
the solution of the Schr{\"o}dinger equation $\big(- \Delta_{disc,h} + V_{disc,h} - E\big) u = f$ on the triangular lattice converges to the solution of the Schr{\"o}dinger equation
$$
\Big(- \frac{1}{3}\big(\frac{\partial^2}{\partial x_1^2} - \frac{\partial^2}{\partial x_1\partial x_2} + \frac{\partial^2}{\partial x_2^2}\big) + V(x) - E\Big)u = f \quad {\it in} \quad {\bf R}^2.
$$
\end{theorem}
Theorem~\ref{ConvergnceofS-matrix} also holds for this case. In this connection we note that the thresholds associated with $P_h(\xi)$ are $0$, $4/(3h^2)$, and $3/(2h^2)$. Thus in the proof of Theorem~\ref{ConvergnceofS-matrix} we have a restriction to the energy interval $(0,4/(3h^2))$. In the limit $h\to0$ this restriction disappears, as for the square lattice. The same remark is valid for the following all examples. Therefore, in the theorems for the Schr{\"o}dinger  limit as in Theorems \ref{S7Schroedingerlimit} and \ref{ConvergnceofS-matrix} to be given below, we do not mention this limitation for $E$, i.e. $0 < E < C_1(h)$, $C_1(h)$ being the first threshold in the spectrum of $\mathcal L_h(S_h)$.


\subsection{Ladder of square lattice}
\begin{figure}[hbtp]
\centering
\includegraphics[width=8cm]{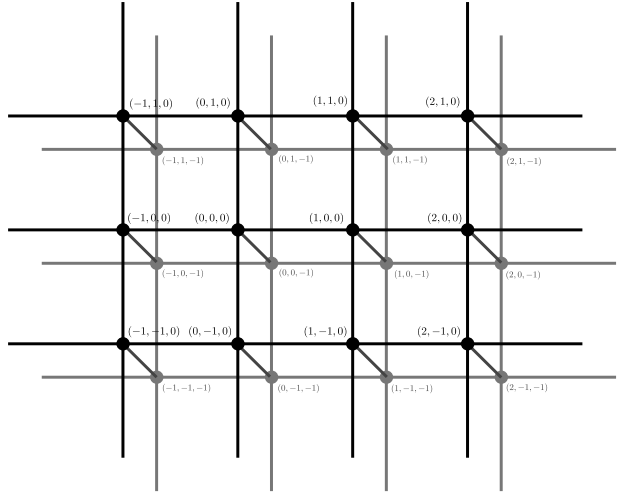}
\caption{2-dim. ladder}
\label{C3-dim.ladder}
\end{figure}
In this case the Laplacian is written as
\begin{equation}
\mathcal L_h(S_h) = \frac{1}{h^2}\mathcal L(S_h),
\nonumber
\end{equation}
\begin{equation}
\mathcal L(S_h) = - \frac{1}{2d+1}
\left(
\begin{array}{cc}
\sum_{j=1}^d\big(S_{h,j} + S_{h,j}^{\ast}\big)& 1 \\
1 & \sum_{j=1}^d\big(S_{h,j} +  S_{h,j}^{\ast}\big)
\end{array}
\right).
\nonumber
\end{equation}
Letting
\begin{equation}
T = \frac{1}{\sqrt2}\left(
\begin{array}{cc}
1 & -1  \\
1 & 1
\end{array}
\right),
\nonumber
\end{equation}
we have
\begin{equation}
T^{\ast}\mathcal L(S_h)T   = - \frac{1}{2d+1}
\left(
\begin{array}{cc}
\sum_{j=1}^d\big(S_{h,j} + S_{h,j}^{\ast}\big) + 1 & 0 \\
0 & \sum_{j=1}^d\big(S_{h,j} + S_{h,j}^{\ast})-1
\end{array}
\right).
\nonumber
\end{equation}
Then the reference energy is $E_0 = - 1/h^2$. 
We then consider the Hamiltonian
\begin{equation}
\mathcal L_h(S_h) =
\frac{1}{(2d+1)h^2}\left(
\begin{array}{cc}
2d - \sum_{j=1}^d\big(S_{h,j} + S_{h,j}^{\ast}\big) & 0 \\
0 & 2d + 2  - \sum_{j=1}^d\big(S_{h,j} + S_{h,j}^{\ast}\big)
\end{array}
\right).
\nonumber
\end{equation}
The characteristic roots are
$$
\lambda^{(+)}_h(\eta) = \frac{2d + 2 - 2\sum_{j=1}^d\cos h\eta_j}{(2d+1)h^2},
$$
$$
\lambda^{(-)}_h(\eta) = \frac{2d  - 2\sum_{j=1}^d\cos h\eta_j}{(2d+1)h^2}.
$$
Letting $P_h(\eta) = \lambda_h^{(-)}(\eta)$ we then have
$$
P_h(\eta) = \frac{4}{(2d+2)h^2}\sum_{j=1}^d\sin^2\frac{\eta_j^2}{2} \to 
\frac{1}{2d+1}\sum_{j=1}^d\eta_j^2.
$$
We then have the following theorem.


 \begin{theorem}
 Let $V(x) \in H^s({\bf R}^d)$ with $s > d/2$ and compactly supported, and $f \in H^{m,s}({\bf R}^d)$ with $m > d/2 + 1$, $s > d+1$. Then
  the solution of the equation $(- \Delta_{disc,h} + V_{disc,h} - E)u = f$ converges to the solution of 
the Schr{\"o}dinger equation
$$
\Big(- \frac{1}{2d+1}\Delta + V(x) - E\Big)u = f \quad {\it in} \quad {\bf R}^d.
$$
\end{theorem}

Theorem \ref{ConvergnceofS-matrix} also holds for this case.

\section{Derivation of Dirac equations}

\subsection{Hexagonal lattice}
\begin{figure}[hbtp]
\centering
\includegraphics[width=8cm,height=6cm,keepaspectratio]{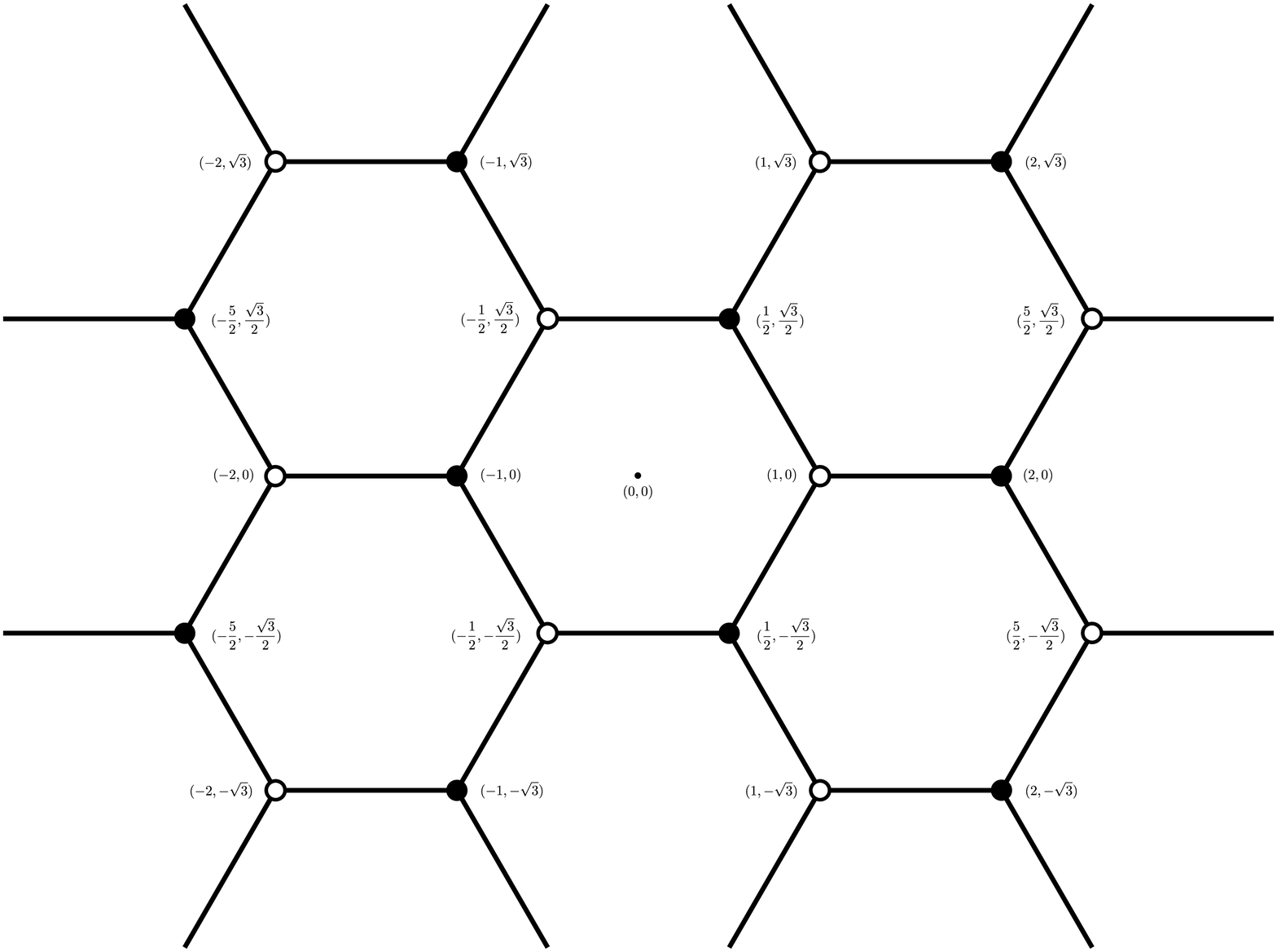}
\caption{Hexagonal lattice}
\label{Fig:HexagonalLattice}
\end{figure} 

The Laplacian on the hexagonal lattice is 
\begin{equation}
 = -
\frac{1}{3}
\left(
\begin{array}{cc}
0 & 1 + S_{1}^{\ast} + S_{2}^{\ast} \\
1 + S_{1} + S_{2}& 0
\end{array}
\right).
\label{LaplacianHexalattice}
\end{equation} 
See Figure  \ref{Fig:HexagonalLattice}.
We put
\begin{equation}
\mathcal L_h(e^{-ih\eta})
  = -
\frac{1}{3h}
\left(
\begin{array}{cc}
0 & 1 + e^{ih\eta_1} + e^{ih\eta_2} \\
 1 + e^{-ih\eta_1} + e^{-ih\eta_2} & 0
\end{array}
\right), 
\label{eigenvauesHexalattice}
\end{equation}
and let $\mathcal L(e^{-\eta}) = \mathcal L_1(e^{-i\eta})$.
Letting
\begin{equation}
b(z) = \cos z_1 + \cos z_2 + \cos(z_1 - z_2), \quad
z = (z_1, z_2) \in {\bf C}^2,
\label{S8b(z)define}
\end{equation}
we have
$$
|1 + e^{i\eta_1} + e^{i\eta_2}|^2 = 3 + 2b(\eta), \quad \eta = (\eta_1,\eta_2) \in {\bf R}^2.
$$
The characteristic roots of $\mathcal L(e^{-i\eta})$ are given by
$\displaystyle{
\lambda^{(\pm)}(\eta) = \pm \frac{\sqrt{3+2b(\eta)}}{3}.}
$
By elementary geometry, we have
\begin{equation}
\begin{split}
3 + 2b(\eta) = 0 &\Longleftrightarrow \eta = (\frac{2\pi}{3},-\frac{2\pi}{3}), \  (-\frac{2\pi}{3},\frac{2\pi}{3}), 
\\
3 + 2b(\eta) = 9 & \Longleftrightarrow \eta = (0,0).
\end{split}
\label{3+2bz=0}
\end{equation}
For the first two cases, the Hessian matrix of $3 + 2b(\eta)$ is 
 $\displaystyle{
\left(
\begin{array}{cc}
1 & -1/2 \\
- 1/2 & 1
\end{array}
\right)}$, and for the third case  $\displaystyle{
\left(
\begin{array}{cc}
1 & -1/\sqrt{2} \\
- 1/\sqrt{2} & 1
\end{array}
\right)}$. 
The characteristic roots of $\mathcal L_h(S_h)$ are
$$
\lambda_h^{(\pm)}(\eta) = \pm \frac{\sqrt{3 + 2b(h\eta)}}{3h}.
$$
The spectrum of $- \frac{1}{h}\Delta_{\Gamma_h}$ is 
$$
\sigma\big( - \frac{1}{h}\Delta_{\Gamma_h}\big) = \big[-1/h,1/h\big].
$$

\subsubsection{Expansion around the Dirac points} 
We take the reference energy $E_0 = 0$. 
We put
\begin{equation}
d_h^{(\pm)} = \pm\frac{1}{h}(\frac{2\pi}{3},-\frac{2\pi}{3}),
\label{S8Diracpoints}
\end{equation}
\begin{equation}
\mathcal K_0^{(\pm)} = \{\eta \in {\bf T}^d\, ; \, |\eta - d_1^{(\pm)}| < \frac{\pi}{3}\},
\nonumber
\end{equation}
\begin{equation}
\mathcal K_0^{(0)} = {\bf T}^d \setminus \big(\mathcal K_0^{(+)}\cup \mathcal K_0^{(-)}\big).
\nonumber
\end{equation}
Take $\chi^{(\pm)}, \chi^{(0)} \in C^{\infty}({\bf T}^d)$ such that 
\begin{equation}
\chi^{(\pm)}(\eta) = 
\left\{
\begin{split}
&1, \quad {\rm for} \quad |\eta - d_1^{(\pm)}| < \frac{\pi}{6}, \\
&0, \quad {\rm for} \quad |\eta - d_1^{(\pm)}| > \frac{\pi}{3},
\end{split}
 \right.
 \nonumber
\end{equation}
\begin{equation}
\chi^{(0)}(\eta) = 1 - \chi^{(+)}(\eta) - \chi^{(-)}(\eta).
\nonumber
\end{equation}
We consider the Schr{\"o}dinger equation
\begin{equation}
(- \Delta_{\Gamma_h}  - z)u_h = f_h
\label{S8EquationHexa}
\end{equation}
on the hexagonal lattice.
Let 
 \begin{equation}
 \begin{split}
&  u^{(\pm)}_h =  \mathcal F_{disc,h}^{-1}\big(\chi^{(\pm)}(h\eta)\widehat u_h(\eta)\big), \\
 & u^{(0)}_h =  \mathcal F_{disc,h}^{-1}\big(\chi^{(0)}(h\eta)\widehat u_h(\eta)\big), \\
 & f^{(\pm)}_h =  \mathcal F_{disc,h}^{-1}\big(\chi^{(\pm)}(h\eta)\widehat f_h(\eta)\big).
 \end{split}
  \nonumber
 \end{equation}
Take a compact interval $I \subset (0,\infty)$ and assume that ${\rm Re}\, z \in I$. 
We  consider (\ref{S8EquationHexa}) on ${\bf T}^d_h$.  
Since $\lambda^{(\pm)}(\eta) \neq 0$ on $\mathcal K_0^{(0)}$, there exists an $\epsilon_0 >0$ such that
$$
|\mathcal \lambda^{(\pm)}_h(\eta)| \geq \frac{\epsilon_0}{h}, \quad 
{\rm on} \quad \mathcal K^{(0)}_0/h.
$$
Therefore there exists $h_0 > 0$ such that for $0 < h < h_0$
$$
\big|\det\big(\mathcal L_h(e^{-h\eta}) - z\big)\big|
 \geq C/h^2 \quad {\rm on} \quad 
 {\rm supp}\, \chi^{(0)}(h\eta)
$$
for a constant $C > 0$. We then have, letting $\|\cdot\|_s = \|\cdot\|_{L^{2,s}}$, 
\begin{equation}
\|u_h^{(0)}\|_{s}\leq 
Ch\big(\|f_h\|_{s} + \|u_h\|_{-s}\big), \quad 
s > 1/2.
\label{S8chi0estimate}
\end{equation}

\medskip
We put
\begin{equation}
P_h^{(\pm)}(\xi) = \frac{\sqrt{3 + 2b\big(h(\xi + d_h^{(\pm)})\big)}}{3h},
\nonumber
\end{equation}
\begin{equation}
\mathcal K^{(\pm)} = \mathcal K_0^{(\pm)} - d_1^{(\pm)} = 
\{\xi = \eta - d_1^{(\pm)}\, ; \, 
\eta \in \mathcal K^{(\pm)}_0\}.
\nonumber
\end{equation}
In view of (\ref{S8b(z)define}), (\ref{3+2bz=0}) and Taylor expansion, we have the following lemma.


\begin{lemma}
\label{LemmaPpmhgeqCxiHexa}
On $\mathcal K^{(\pm)}/h$, $P_h^{(\pm)}(\xi)$ vanishes only at $\xi = 0$. 
Moreover,  there exists a constant $C > 0$ independent of $h$ such that
\begin{equation}
P_h^{(\pm)}(\xi) \geq C|\xi|, \quad \xi \in \mathcal K^{(\pm)}/h.
\nonumber
\end{equation}
\end{lemma}

In Subsection \ref{S4.5Dirac}, we studied this equation by projecting 
onto each characteristic root. 
Here, we deal with $\mathcal L_h(S_h)$ 
without diagonalization. For the solution of the equation
\begin{equation}
 \big(- \Delta_{disc,h} - z\big)u_h = f_h,
 \label{HexagonalLatticeEq}
\end{equation}
  we split
  $u_h = u_h^{(+)} + u_h^{(-)}$, where
 \begin{equation}
 u^{(\pm)}_h =  \mathcal F_{disc,h}^{-1}\big(\chi^{(\pm)}(h\eta)\widehat u_h(\eta)\big), \quad 
  f^{(\pm)}_h =  \mathcal F_{disc,h}^{-1}\big(\chi^{(\pm)}(h\eta)\widehat f_h(\eta)\big),
  \nonumber
 \end{equation}
 which satisfy
 \begin{equation}
 \big(\mathcal L_h(S_h) - z\big)u^{(\pm)}_h = f^{(\pm)}_h.
 \nonumber
 \end{equation}
Define the gauge transformation $\mathcal G^{(\pm)}$ by
\begin{equation}
\big(\mathcal G^{(\pm)} a\big)(n) = e^{ihn\cdot d_h^{(\pm)}}a(n) =  e^{in\cdot d_1^{(\pm)}}a(n), \quad 
a \in L^2({\bf Z}^2_h).
\nonumber
\end{equation}
We put
\begin{equation}
v^{(\pm)}_h = \big(\mathcal G^{(\pm)}\big)^{\ast}u^{(\pm)}_h,
\nonumber
\end{equation}
and consider
\begin{equation}
\big(\big(\mathcal G^{(\pm)}\big)^{\ast}\mathcal L_h(S_h)\mathcal G^{(\pm)} - z\big)v^{(\pm)}_h = \big(\mathcal G^{(\pm)}\big)^{\ast}f^{(\pm)}_h.
\nonumber
\end{equation}
Note that
\begin{equation}
\big(\mathcal F_{disc,h}\big(\mathcal G^{(\pm)}\big)^{\ast}\mathcal L_h(S_h)\mathcal G^{(\pm)} a\big)(\xi) = \mathcal L_{h}(e^{ih(\xi + d_h^{(\pm)})})\big(\mathcal F_{disc,h}a\big)(\xi).
\end{equation}
We put
$$
q(\eta) = 1 + e^{i\eta_1} + e^{i\eta_2}.
$$
In view of (\ref{eigenvauesHexalattice}), we have
\begin{equation}
\Big(\mathcal L_{h}(e^{ih(\xi + d_h^{(\pm)}}) - z\Big)^{-1} = 
\frac{1}{P_h^{(\pm)}(\xi) - z^2}\left(
\begin{array}{cc}
0 & q(h(\xi + d_h^{(\pm)}))\\
q(- h(\xi + d_h^{(\pm)})) & 0
\end{array}
\right).
\nonumber
\end{equation}
By virtue of Lemma \ref{LemmaPpmhgeqCxiHexa}, one can argue in the same way as in \S \ref{Section4Freeequation} to obtain the uniform estimates with respect to $0 < h < h_0$.
We take $\psi^{(\pm)}(\eta) \in C_0^{\infty}({\bf R}^d)$ and put
$$
f_h(x) = \big(\frac{h}{2\pi}\big)^{d/2}\int_{{\bf R}^d}
e^{ix\cdot\eta}
\left(\psi^{(+)}(\eta - d_h^{(+)}) + \psi^{(-)}(\eta - d_h^{(-)})\right)d\eta.
$$

\begin{lemma}
\label{Lemma8.2}
Let $v^{(\pm)}_h = v^{(\pm)}_h( E + i0)$ for $E > 0$. Then we have for $\epsilon > 0$
\begin{equation}
\|\widetilde v^{(\pm)}_h\|_{-1/2 - \epsilon} \leq C\|f\|_{m,s},
\nonumber
\end{equation}
\begin{equation}
\|p_-(D_x)\widetilde v_h^{(\pm)}\|_{-1/2 + \epsilon} \leq 
C\|f\|_{m,s}, \quad 
p_- \in \mathcal P_-.
\nonumber
\end{equation}
\end{lemma}

\medskip
Let $u_h = u_h(E + i0) = \mathcal G^{(+)}v^{(+)} + \mathcal G^{(-)}v^{(-)}$. Noting that 
$$
\widehat{\mathcal G^{(\pm)}v^{(\pm)} } = \widehat v^{(\pm)}(\xi - d_h^{(\pm)}),
$$
we have
$$
\widetilde u_h = e^{ix\cdot d_h^{(+)}}\widetilde v_h^{(+)} + e^{ix\cdot d_h^{(-)}}\widetilde v_h^{(-)}.
$$
Lemma \ref{Lemma8.2} yields
$$
\|\widetilde u_h\|_{-1/2-\epsilon} \leq C\|f\|_{m,s}.
$$
By the same argument as in \S \ref{Section4Freeequation}, we see that
$v_h^{(\pm)} \to v^{(\pm)}$, hence $\widetilde u_h$ behaves like
$$
\widetilde u_h \simeq e^{ix\cdot d_h^{(+)}}\widetilde v^{(+)} + e^{ix\cdot d_h^{(-)}}\widetilde v^{(-)}.
$$
We show that $\widetilde v^{(\pm)}$ are solutions to massless Dirac equations.

We consider the case of $v_h^{(+)}$, and  make the change of variables
$$
\eta_1 = \xi_1 + \frac{2\pi}{3h}, \quad 
\eta_2 = \xi_2 - \frac{2\pi}{3h},
$$
to obtain
\begin{equation}
\begin{split}
1 + e^{ih\eta_1} + e^{ih\eta_2} &= e^{2\pi i/3}\big(e^{ih\xi_1} - 1\big) +  e^{-2\pi i/3}\big(e^{ih\xi_2} - 1\big) \\
& \sim -hi\frac{\xi_1 + \xi_2}{2} - h\frac{\sqrt{3}(\xi_1-\xi_2)}{2}
\end{split}
\nonumber
\end{equation}
as $h \to 0$. 
We put 
$$
\zeta_1 = \frac{\sqrt{3}}{6}\big(\xi_1 - \xi_2\big), \quad \zeta_2 = - \frac{1}{6}\big(\xi_1 + \xi_2\big).
$$
Then we have
\begin{equation}
- \frac{1}{3h}\big(1 + e^{ih\eta_1} + e^{ih\eta_2}\big) \sim \zeta_1 - i\zeta_2.
\nonumber
\end{equation}
We put
$$
y_1 = \sqrt{3}\big(x_1 - x_2\big), \quad y_2 = -3\big(x_1+x_2\big).
$$
Then the map $(x,\eta) \to (y,\zeta)$ is a symplectic transformation. We then have
\begin{equation}
-\frac{1}{3h}\big(1 + e^{ih\eta_1} + e^{ih\eta_2}\big) \sim 
\zeta_1 - i\zeta_2,
\nonumber
\end{equation}
as $h \to 0$.
Similarly
\begin{equation}
-\frac{1}{3h}\big(1 + e^{-ih\eta_1} + e^{-ih\eta_2}\big) \sim 
\zeta_1 + i\zeta_2. 
\nonumber
\end{equation}
We have thus obtained
\begin{equation}
\mathcal L_h(e^{ih(\xi + h_d)}) \sim
\left(
\begin{array}{cc}
0 & \zeta_1 - i\zeta_2 \\
\zeta_1 + i\zeta_2 & 0
\end{array}
\right)
= \zeta_1\sigma_1 + \zeta_2\sigma_2,
\nonumber
\end{equation}
where $\sigma_1$, $\sigma_2$ are Pauli spin  matrices.
Therefore, $v^{(+)}$ satisfies
\begin{equation}
\Big(\sigma_1\frac{1}{i}\frac{\partial}{\partial y_1} +  \sigma_2\frac{1}{i}\frac{\partial}{\partial y_2} - E\Big)v^{(+)} = g^{(+)}.
\label{Diraceq+}
\end{equation}
Similarly $v^{(-)}$ satisfies
\begin{equation}
\Big( -\sigma_1\frac{1}{i}\frac{\partial}{\partial y_1} +  \sigma_2\frac{1}{i}\frac{\partial}{\partial y_2} - E\Big)v^{(-)} = g^{(-)},
\label{Diraceq-}
\end{equation}
where
$$
g^{(\pm)} = \mathcal F_{cont}^{-1}\, \psi^{(\pm)}.
$$
We have thus proven the following theorem.

\begin{theorem}
Assume that $f \in H^{m,s}({\bf R}^2)$ with $m > 2$, $s > 3$.  
Then for the solution $u_h$ of (\ref
{HexagonalLatticeEq}),  $\widetilde u_h$ behaves like
$$
\widetilde u_h \simeq e^{ix\cdot d_h^{(+)}}\widetilde v^{(+)} + e^{ix\cdot d_h^{(-)}}\widetilde v^{(-)}.
$$
Moreover, $\widetilde v^{(\pm)}$ satisfy the massless Dirac equation (\ref{Diraceq+}) and (\ref{Diraceq-}) after the symplectic transformation.
\end{theorem}

\subsubsection{Global minimum}
To deal with the case near the lowest energy, instead of  (\ref{eigenvauesHexalattice}),  we should consider
\begin{equation}
 -
\frac{1}{3h^2}
\left(
\begin{array}{cc}
0 & 1 + S_{h,1}^{\ast} + S_{h,2}^{\ast} \\
1 + S_{h,1} + S_{h,2}& 0
\end{array}
\right).
\nonumber
\end{equation}
The reference energy is $E_0 = - h^2/2$, and we consider
the Hamitonian 
\begin{equation}
\mathcal L_h(e^{-ih\eta}) 
  = -
\frac{1}{3h^2}
\left(
\begin{array}{cc}
-3 & 1 + e^{ih\eta_1} + e^{ih\eta_2} \\
 1 + e^{-ih\eta_1} + e^{-ih\eta_2} & -3
\end{array}
\right).
\label{S8reducedHamiltonian}
\end{equation}
Then the characteristic roots are
$$
\lambda_h^{(\pm)}(\eta) =\frac{3 \pm \sqrt{3 + 2b_2(h\eta)}}{3h^2}.
$$
The minimum is attained only at $\eta = 0$, and $\lambda_h^{(-)}(\eta)$ has a Taylor expansion
$$
\lambda_h^{(-)}(\eta) = \frac{1}{9}\left(\eta_1^2 - \eta_1\eta_2 + \eta_2^2\right) + O(h^2).
$$
Therefore we obtain the following theorem.

\begin{theorem}
Assume that $f \in H^{m,s}({\bf R}^2)$ with $m > 2$, $s > 3$.  Assume also that $V \in H^2({\bf R}^2)$ with $s > 1$, and $z \not\in {\bf R}$. 
Then the solution of the Schr{\"o}dinger equation 
$(- \Delta_{disc,h} + V_{disc,h} - z)u_h = f_h$ 
on the hexagonal lattice, where $\Delta_{disc,h}$ is the difference operator with symbol (\ref{S8reducedHamiltonian}),
converges to that for the continuum Schr{\"o}dinger equation
$$
\Big(-\frac{1}{9}\big(\frac{\partial^2}{\partial x_1^2} -\frac{ \partial^2}{\partial x_1\partial x_2} + \frac{\partial^2}{\partial x_2^2}\big)  + V(x) - z\Big)u = f, \quad {\it in} \quad {\bf R}^2.
$$
\end{theorem}

\subsection{Graphite}
The Laplacian for  graphite is written as
\begin{equation}
H_{0h}(S_h,S_h^{\ast}) = \frac{1}{h}H_0(S_h,S_h^{\ast}),
\nonumber
\end{equation}
\begin{equation}
H_0(S_h,S_h^{\ast}) = - \frac{1}{4}
\left(
\begin{array}{cccc}
0 & 1 + S_{h,1} + S_{h,2} & 1 & 0 \\
1 + S_{h,1}^{\ast} + S_{h,2}^{\ast} & 0 & 0 & 1\\
1 & 0 & 0 & 1 + S_{h,1} + S_{h,2}\\
0 & 1 & 1 + S_{h,1}^{\ast} + S_{h,2}^{\ast} & 0
\end{array}
\right)
\nonumber
\end{equation}
Put
\begin{equation}
T = \frac{1}{\sqrt{2}}\left(
\begin{array}{cc}
I_2 & - I_2 \\
I_2 & I_2
\end{array}
\right), \quad I_2 = \left(
\begin{array}{cc}
1 & - 1\\
1 & 1
\end{array}
\right). 
\nonumber
\end{equation}
Then it can be block diagonalized as follows
\begin{equation}
T^{\ast}H_0(S_h,S_h^{\ast})T = - \frac{1}{4}
\left(
\begin{array}{cc}
C(S_h,S_h^{\ast}) + I_2 & 0 \\
0 & C(S_h,S_h^{\ast}) - I_2
\end{array}
\right)
\nonumber
\end{equation}
\begin{equation}
C(S_h,S_h^{\ast})  = 
\left(
\begin{array}{cc}
0 & 1 + S_{h,1} + S_{h,2} \\
 1 + S_{h,1}^{\ast} + S_{h,2}^{\ast} & 0
\end{array}
\right).
\nonumber
\end{equation}
Therefore it can be dealt with in the same way as in the hexagonal lattice.

\begin{figure}
\centering
\includegraphics[width=8cm]{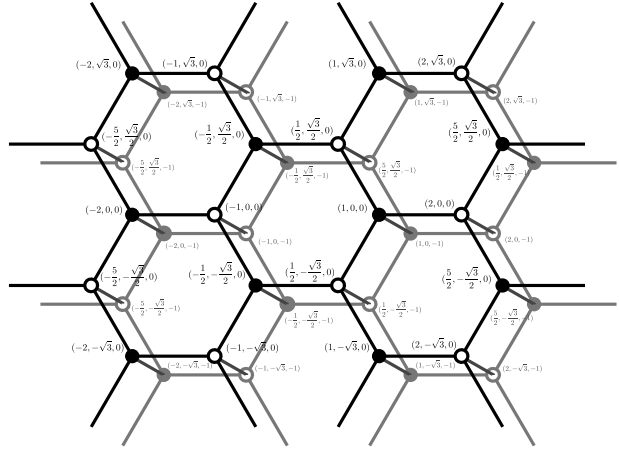}
\label{C1Graphite}
\caption{Graphite}
\end{figure}

\section{Other examples}

\subsection{Kagome lattice}\label{C1kagomeLattice}

Passing to the Fourier series, the symbol of $-\Delta_{\Gamma}$ for the Kagome lattice is written as
\begin{equation}
\mathcal L(e^{-i\eta}) =- \frac{1}{4} 
\left(
\begin{array}{ccc}
0& 1 + e^{i\eta_1}e^{-i\eta_2} &1 + e^{i\eta_1}\\
1+ e^{-i\eta_1}e^{i\eta_2}& 0& 1 + e^{i\eta_2}\\
1 + e^{-i\eta_1} &1 + e^{-i\eta_2}& 0
\end{array}
\right).
\nonumber
\end{equation}
The characteristic determinant is
\begin{equation}
p(\eta,\lambda) = \det\big(\mathcal L(e^{-i\eta})- \lambda\big) = -\big(\lambda - \frac{1}{2}\big)\big(\lambda^2 + \frac{\lambda}{2} - \frac{\beta(\eta)}{8}\big),
\nonumber
\end{equation}
\begin{equation}
\beta(\eta) = 1 +  \cos \eta_1 + \cos \eta_2 + \cos(\eta_1-\eta_2).
\nonumber
\end{equation}
The characteristic roots are
\begin{equation}
\lambda = - \frac{1}{4} \pm \frac{\sqrt{2\beta(\eta) + 1}}{4}.
\label{Kagomeroots}
\end{equation}
Then the spectrum of $- \Delta_{\Gamma}$ is  $\sigma(- \Delta_{\Gamma}) = [-1,1/2]$. 

We first take the reference energy $E_0 = - 1/h^2$, and consider 
$$
\mathcal L_h(e^{-i\eta}) = -\frac{1}{4h^2}\left(
\begin{array}{ccc}
-4& 1 + e^{ih\eta_1}e^{-ih\eta_2} &1 + e^{ih\eta_1}\\
1+ e^{-ih\eta_1}e^{ih\eta_2}& -4& 1 + e^{ih\eta_2}\\
1 + e^{-ih\eta_1} &1 + e^{-ih\eta_2}& -4
\end{array}
\right).
$$
\begin{figure}[hbtp]
\centering
\includegraphics[width=7cm]{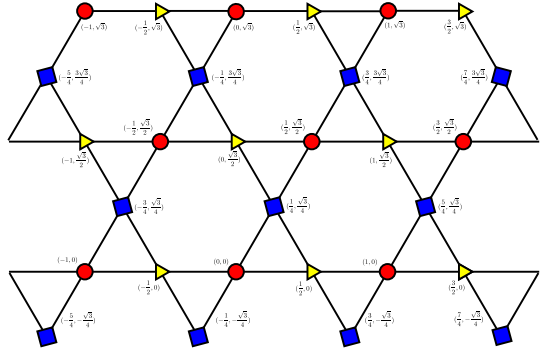}
\caption{Kagome lattice}
\label{C1KagomeLattice}
\end{figure}
The least characteristic root is then given by
$$
P_h(\eta) = \frac{3 - \sqrt{2\beta(h\eta) + 1}}{4h^2}.
 $$ 
 It is expanded as
 \begin{equation}
 \begin{split}
 P_h(\eta) &= \frac{1}{6h^2}\Big(\sin^2\frac{h\eta_1}{2} + 
 \sin^2\frac{h\eta_2}{2} + \sin^2\frac{h(\eta_1-\eta)^2}{2}\Big) \\
 &= \frac{1}{12}\Big(\eta_1^2 - \eta_1\eta_2 + \eta_2^2\Big) + O(h^2).
 \end{split}
 \nonumber
 \end{equation}
For the Kagome lattice the unique continuation theorem does not hold, and  there may be embedded eigenvalues. Therefore we consider the continuum limit for the complex energy $z \not\in {\bf R}$. 

\begin{theorem}
Assume that $f \in H^{m,s}({\bf R}^d)$ with $s >3$ and $m > 2$, and that $z \not\in {\bf R}$. 
Assume also that $V \in H^s({\bf R}^2)$ with $s > 1$.
Then  the solution of $(- \Delta_{disc,h} + V_{disc,h} -z)u = f$ converges to that for the Schr{\"o}dinger equation
$$
\Big(- \frac{1}{12}\big(\frac{\partial^2}{\partial x_1^2} - \frac{ \partial^2}{\partial x_1\partial x_2} + \frac{\partial^2}{\partial x_2^2}\big) +V(x)  - z\Big)v = g.
$$
\end{theorem}

For $\lambda = - 1/4$, the characteristic toots (\ref{Kagomeroots}) are double. In this case, we take the reference energy $E_0 = - 1/4$ and consider

$$
\mathcal L_h(e^{-i\eta}) = -\frac{1}{4h}\left(
\begin{array}{ccc}
-1& 1 + e^{ih\eta_1}e^{-ih\eta_2} &1 + e^{ih\eta_1}\\
1+ e^{-ih\eta_1}e^{ih\eta_2}& -1& 1 + e^{ih\eta_2}\\
1 + e^{-ih\eta_1} &1 + e^{-ih\eta_2}& -1
\end{array}
\right).
$$
Then $2\beta(h\eta) + 1 = 0$ if and only if $\eta$ is the Dirac point, i.e. $\eta = d_h$, where $d_h$ is as in (\ref{S8Diracpoints}). One can then argue as in the case of the hexagonal lattice to show the following theorem.

 \begin{theorem}
Assume that $f \in H^{m,s}({\bf R}^d)$ with $s >3$ and $m > 2$, and that $z \not\in {\bf R}$.
Then  the solution of $\big(- \mathcal G_h^{\ast} \Delta_{disc,h}\mathcal G_h -z\big)u = f$ behaves like
$$
\widetilde u_h \simeq e^{ix\cdot d_h^{(+)}}\widetilde v^{(+)} + e^{ix\cdot d_h^{(-)}}\widetilde v^{(-)},
$$
and $\widetilde v^{(\pm)}$ satisfy the massless Dirac equation (\ref{Diraceq+}), (\ref{Diraceq-}).
\end{theorem} 

\subsection{Subdivision of square lattice}

\begin{figure}[hbtp]
\centering
\includegraphics[width=6cm]{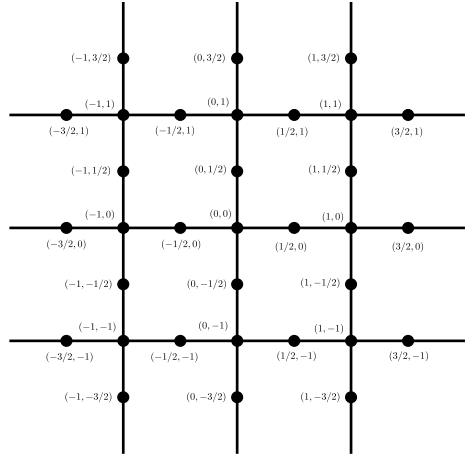}
\caption{Subdivision of $2$-dimensional square lattice}
\label{C1Subdivision}
\end{figure}

For this lattice the unique continuation property does not hold (see \cite{AndIsoMor1}). Therefore we apply Theorems \ref{TheoremwidetildediscretecomplexzB} and \ref{TheoremwidetildediscretecomplexzC}.

Passing to the Fourier series in the case $h=1$, $- \widehat\Delta_{disc}$ becomes the following matrix
\begin{equation}
\mathcal L(e^{-i\eta}) = -\frac{1}{2\sqrt{d}}
\left(
\begin{array}{cccc}
0 & 1 + e^{-i\eta_1} & \cdots & 1+e^{-i\eta_d}\\
1 + e^{i\eta_1} & 0 & \cdots & 0 \\
\vdots & \vdots & \ddots & \vdots \\
1 + e^{i\eta_d} & 0 & \cdots & 0
\end{array}
\right),
\nonumber
\end{equation}
whose determinant is computed as
\begin{equation}
\det(\mathcal L(e^{-\eta}) - \lambda) =  (-\lambda)^{d-1} \big(\lambda^2 - \frac{1}{2d}
(d + \sum_{j=1}^d\cos \eta_j)\big).
\label{Subdivisionpxlambda}
\end{equation}
Therefore the characteristic roots are
$$
\lambda^{(\pm)}(\eta) = \pm\sqrt{\frac{1}{2d}
(d + \sum_{j=1}^d\cos \eta_j)}
$$
To observe the behavior near the bottom of energies, we take the reference energy $E_0 = - 1/h^2$, and consider
\begin{equation}
\mathcal L_h(e^{-ih\eta}) = -\frac{1}{2\sqrt{d}h^2}
\left(
\begin{array}{cccc}
- 2\sqrt{d} & 1 + e^{-ih\eta_1} & \cdots & 1+e^{-ih\eta_d}\\
1 + e^{ih\eta_1} & - 2\sqrt{d}  & \cdots & 0 \\
\vdots & \vdots & \ddots & \vdots \\
1 + e^{ih\eta_d} & 0 & \cdots & - 2\sqrt{d} 
\end{array}
\nonumber
\right).
\end{equation}
We put
$$
\lambda_h^{(\pm)}(\eta) = 
\frac{1}{h^2}\Big(1 \pm\sqrt{\frac{1}{2d}
(d + \sum_{j=1}^d\cos h\eta_j)\Big)}
$$
As $h \to 0$, it behaves like
$$
\lambda^{(-)}_h(\eta) = \frac{1}{8}\sum_{j=1}^d\eta_j^2 + O(h^2).
$$

\begin{theorem}
Let $V(x) \in H^s({\bf R}^d)$ with $s > d/2$ and compactly supported,
Assume that $f \in H^{m,s}({\bf R}^d)$ with $s > d + 1$ and $m > d/2 + 1$. If $z \not\in {\bf R}$, the solution of the equation $( -\Delta_{disc,h} + V_{disc,h} - z)u = f$ converges to the solution of 
$$
\big(- \frac{1}{8}\Delta + V(x) - z\big)v = g,
$$
where $v$ and $g$ are given in Theorem \ref{TheoremwidetildediscretecomplexzB}. 
\end{theorem}

To study the behavior near the 0 energy of $- \Delta_{\Gamma}$, we take the reference energy $E_0 = 0$, and consider
\begin{equation}
\mathcal L_h(e^{-ih\eta}) = -\frac{1}{2\sqrt{d}h}
\left(
\begin{array}{cccc}
0 & 1 + e^{-ih\eta_1} & \cdots & 1+e^{-ih\eta_d}\\
1 + e^{ih\eta_1} & 0  & \cdots & 0 \\
\vdots & \vdots & \ddots & \vdots \\
1 + e^{ih\eta_d} & 0 & \cdots & 0
\end{array}
\right).
\nonumber
\end{equation}
Then the characteristic roots are
$$
\lambda^{(\pm)}(\eta) = \pm \frac{1}{h}\sqrt{\frac{1}{2d}
(d + \sum_{j=1}^d\cos h\eta_j)}.
$$
They vanish if and only if $\eta = d_h$, where
$$
d_h = \frac{1}{h}(\pi,\cdots,\pi).
$$
We then have 
$$
\lambda_h^{(\pm)}(\xi + d_h) = |\xi| + O(h^2).
$$

\begin{theorem}
Let $z \not\in {\bf R}$, and assume that $f \in H^{m,s}({\bf R}^d)$ with $s > d + 1$ and $m > d/2 + 1$. Assume also that $V \in H^s({\bf R}^2)$ with $s > 1$. 
Then the solution of the gauge transformed equation $(\mathcal - G_h^{\ast} \Delta_{disc,h}\mathcal G_h + V_{disc,h} - z)u = f$ converges to the solution of 
$$
( - |D_x| + V(x) - z)v = g,
$$
where $v$ and $g$ are given in Theorem \ref{TheoremwidetildediscretecomplexzB}. 
\end{theorem}

\section{Appendix}
\subsection{The Riemann integral}
We recall the notation
\begin{equation*}
I_{hn}=hn+[-h/2,h/2]^d,\quad n\in\mathbf{Z}^d.
\end{equation*}
We recall \eqref{S4widehatfhxidefine}. Assume $f\in L^2(\mathbf{R}^d)\cap C(\mathbf{R}^d)$ and $(f(hn))_{n\in\mathbf{Z}^d}\in L^1(\mathbf{Z}^d_h)$. Then
\begin{equation*}
\widehat{f}_h(\xi)=\bigl(\frac{h}{2\pi}\bigr)^{d/2}
\sum_{n\in\mathbf{Z}^d}f(hn)e^{-ihn\cdot\xi}.
\end{equation*}
Recall the weighted Sobolev space (\ref{S2DefineHms}),  whose norm is denoted by $\lVert\cdot\rVert_{m,s}$.


\begin{lemma}
\label{AppendPointwiseConv}
{\rm(1)} Assume $m>[d/2]$ and $s>d/2$. Let $f\in  H^{m,s}(\mathbf{R}^d)$. Then
\begin{equation}
h^{d/2}\lVert\widehat{f}_h\rVert_{L^2(\mathbf{T}^d_h)}
\leq C \lVert f \rVert_{m,s}.
\label{AppendhdfhL2bfTdleqCfalphabetah<h0}
\end{equation}

\noindent
{\rm(2)}
Assume $m>[d/2]$ and $s>d$. Let $f\in  H^{m,s}(\mathbf{R}^d)$. Then
\begin{equation}
h^d\sum_{n\in\mathbf{Z}^d}\lvert f(hn)\rvert
\leq C \lVert f \rVert_{m,s}.
\end{equation}

\noindent
{\rm(3)}
Assume $m>[d/2]+1$ and $s>d$. Let $f\in  H^{m,s}(\mathbf{R}^d)$. Then
\begin{equation}
\Big\lvert h^d\sum_{n\in\mathbf{Z}^d} f(hn) 
-\int_{\mathbf{R}^d}f(x)dx\Big\rvert
\leq C h\, \lVert f \rVert_{m,s}.
\end{equation}
\end{lemma} 
\begin{proof}
Let $hx\in I_{hn}$. Then $hx - hn\in [-h/2,h/2]^d$, which implies
$|hx|\leq |hn| + \frac12\sqrt{d}h$. Thus for $h\in[0,1]$ and $hx\in I_{hn}$ we have 
\begin{equation}
1+|hx|\leq C_d(1+|hn|).
\end{equation}
Assume $s>d$. Then we have
\begin{align}\label{est1}
\sum_{n\in\mathbf{Z}^d}(1+|hn|)^{-s}h^d&=
h^d\sum_{n\in\mathbf{Z}^d}\int_{I_{1n}}(1+|hn|)^{-s}dx\notag\\
&\leq C_dh^d\sum_{n\in\mathbf{Z}^d}\int_{I_{1n}}(1+|hx|)^{-s}dx\notag\\
&=C_dh^d\int_{\mathbf{R}^d}(1+|hx|)^{-s}dx=C<\infty.
\end{align}
If $f\in H^{m,s}(\mathbf{R}^d)$, where $m>[d/2]$ and $s\geq0$, then the Sobolev inequality implies
\begin{equation}\label{Sobolev}
|f(x)|\leq C \lVert f \rVert_{m,s} \langle x\rangle^{-s}.
\end{equation}

Assume $m>[d/2]$ and $s>d/2$. By the Parseval equation, \eqref{Sobolev}, and \eqref{est1} we have
\begin{equation*}
h^{d/2}\lVert\widehat{f}_h\rVert_{L^2(\mathbf{T}^d_h)}^2=
h^d\sum_{n\in\mathbf{Z}^d}|f(hn)|^2
\leq C\lVert f \rVert_{m,s}\sum_{n\in\mathbf{Z}^d}
\langle hn\rangle^{-2s}h^d\leq C\lVert f \rVert_{m,s}.
\end{equation*}
Thus part $(1)$ follows.

Assume $m>[d/2]$ and $s>d$. Then
\begin{equation*}
\sum_{n\in\mathbf{Z}^d}|f(hn)|h^d\leq
\sum_{n\in\mathbf{Z}^d}(1+|hn|)^{-s}h^d\leq C\lVert f \rVert_{m,s}
\end{equation*}
and part $(2)$ follows.

Assume $m>[d/2]+1$ and $s>d$. Let $f\in  H^{m,s}(\mathbf{R}^d)$. 
Then note that
\begin{equation*}
|\partial_x^{\alpha}f(x)|\leq C\lVert f \rVert_{m,s}
\langle x\rangle^{-s},\quad |\alpha|\leq1.
\end{equation*}
Define $g(t)=f(x+t(y-x))$, $x,y\in\mathbf{R}^d$. Then
\begin{equation*}
f(y)-f(x)=\int_0^1g'(t)dt
=\int_0^1(\nabla f)(x+t(y-x))\cdot(y-x)dt.
\end{equation*}
The following estimate holds:
\begin{equation}
|(\nabla f)(x+t(y-x))|
\leq C \lVert f \rVert_{m,s}(1+|x+t(y-x)|)^{-s}.
\end{equation}

Since $t\in[0,1]$, we have for all $x,y\in\mathbf{R}^d$ with $|x-y|\leq1$
\begin{equation*}
|x+t(y-x)|\geq |x|-|t(y-x)|\geq |x|-1,
\end{equation*}
which implies
\begin{equation}
1+|x|\leq 2(1+|x+t(y-x)|).
\end{equation}
Thus for $x,y\in\mathbf{R}^d$ with $|x-y|\leq1$ we have
\begin{equation*}
|(\nabla f)(x+t(y-x))|\leq C \lVert f \rVert_{m,s}
(1+|x|)^{-s}
\end{equation*}
and then
\begin{equation*}
|f(y)-f(x)|\leq C \lVert f \rVert_{m,s}
(1+|x|)^{-s}|x-y|.
\end{equation*}
Assume $y\in I_{hn}$. Then
\begin{equation*}
|f(y)-f(hn)|\leq C \lVert f \rVert_{m,s}
(1+|hn|)^{-s}h.
\end{equation*}
(Note that the above estimates hold with $|x-y|\leq1$ replaced by $|x-y|\leq c_0$ for some fixed $c_0$, depending on $d$.) Integrating we get
\begin{equation*}
h^df(hn)=\int_{I_{hn}}f(y)dy+R(h,n),
\end{equation*}
where
\begin{equation*}
|R(h,n)|\leq  C \lVert f \rVert_{m,s}
(1+|hn|)^{-s}h^{d+1}.
\end{equation*}
Then 
\begin{equation*}
\sum_{n\in\mathbf{Z}^d}h^df(hn)
=\sum_{n\in\mathbf{Z}^d}\int_{I_{hn}}f(y)dy+
\sum_{n\in\mathbf{Z}^d}R(h,n)
\end{equation*}
and 
\begin{equation*}
\sum_{n\in\mathbf{Z}^d}|R(h,n)|\leq 
C \lVert f \rVert_{m,s}
\sum_{n\in\mathbf{Z}^d}(1+|hn|)^{-s}h^{d+1}
\leq C\lVert f \rVert_{m,s}h
\end{equation*}
by \eqref{est1}. This concludes the proof of part $(3)$.
\end{proof}

\subsection{Lemmas for the Besov space}
 We use the following lemmas in  \S \ref{S2Preliminaries} and \S \ref{Section4Freeequation}, which are in \cite{HoVol2},  Chap. 14, \S 1 or follow from an adaptation of the proof there. 

\begin{lemma}
\label{alphabetalemma}
Let $b_n \geq 0$, $n = 0, 1, 2, \cdots$, and put
$$
\alpha = \sup_{n \geq 0}\frac{b_n}{2^n}, \quad 
\beta = \sup_{n\geq 0}\frac{1}{2^n}
\big(b_0 + b_1 + \cdots + b_n\big).
$$
Then we have
$$
\alpha \leq \beta \leq 3\alpha.
$$
\end{lemma}

\begin{lemma}
\label{AB2ALemma}
Letting 
$$
A = \sup_{j\geq 0}\frac{h^d}{2^j}\sum_{|hn|\leq 2^j}|u(n)|^2h^d,
\quad
B = \sup_{R>1}\frac{h^d}{R}\sum_{|hn|\leq R}|u(n)|^2,
$$
we have
$$
A \leq B \leq 2A.
$$
\end{lemma}


\begin{lemma}
\label{Lemmafx1L2d-1leqfBnorm}
(1) We  have
$$
\int_{-\infty}^{\infty}\|f(x_1,\cdot)\|_{L^2({\bf R}^{d-1})} \leq 
\sqrt{2}\|f\|_{\mathcal B({\bf R}^d)},
$$
\begin{equation}
\|f\|_{\mathcal B^{\ast}({\bf R}^d)} \leq \sqrt{2} \sup_{x_1 \in {\bf R}}\|f(x_1,\cdot)\|_{L^2({\bf R}^{d-1})}.
\nonumber
\end{equation}
(2) If $P \in S^0_{1,0}$, we have
$$
P \in {\bf B}({\mathcal B};{\mathcal B}) \cap {\bf B}({\mathcal B}^{\ast};{\mathcal B}^{\ast}) \cap 
{\bf B}({\mathcal B}^{\ast}_0;{\mathcal B}^{\ast}_0).
$$
\end{lemma}

\subsection{A priori estimates for discrete Schr{\"o}dinger equations}
We prove here a priori estimates needed in the proof in \S \ref{SectionPotentialPerturbation}. For the sake of simplicity, we explain the proof for the case of square lattice. It works for the general case. 
Letting
\begin{equation}
D_{h,j} = \frac{1}{h}\left(I - S_{h,j}\right),
\nonumber
\end{equation}
we have
\begin{equation}
- \Delta_{disc,h} = \sum_{j=1}^d D_{h,j}^{\ast}D_{h,j} = \sum_{j=1}^d D_{h,j}D_{h,j}^{\ast}.
\label{Eq:-Deltadh=sumDjDj}
\end{equation}
Take a function $\chi(x)$ on ${\bf R}^d$, and  let $\chi_h$ be the operator of multiplication by the function $\chi(hn)$ on $L^2({\bf Z}^d_h)$.
Then we have
\begin{equation}
\big([D_{h,j},\chi_h]u\big)(n) = \frac{1}{h}\left(\chi(hn) - \chi(hn - he_j)\right)\big(S_{h,j}u\big)(n),
\label{Dhjchihcommutator}
\end{equation}
\begin{equation}
\big([D_{h,j}^{\ast},\chi_h]u\big)(n) = \frac{1}{h}\left(\chi(hn) - \chi(hn + he_j)\right)\big(S_{h,j}^{\ast}u\big)(n).
\label{Dhjastchicommutator}
\end{equation}
We also have
\begin{equation}
[\Delta_{disc,h},\chi_h] = \sum_{j=1}^d\left(D_{h,j}^{\ast}[D_{h,j},\chi_h] + 
[D^{\ast}_{h,j},\chi_h]D_{h,j}\right).
\label{Eq:Deltadischchicomuutator}
\end{equation}

 In the following, constants $C$ are independent of $0 < h < h_0$.

\begin{lemma}
\label{Lemmachidifference}
Let $\chi(x) = (1 + |x|^2)^{1/2}$. Then for any $s > 0$, 
\begin{equation}
\Big|\frac{1}{h}\big(\chi^{-s}(hn) - \chi^{-s}(h(n - y))\big)\Big| \leq C_s\chi^{-(s+1)}(hn), \ \  {for} \ \  |y| \leq 1, \ \
n \in {\bf Z}^d, \ \  0 < h < 1/2.
\nonumber
\end{equation}
\end{lemma}

\begin{proof} 
Letting 
$
g(t) = \chi^{-s}(thn + (1-t)h(n-y))$, and noting that 
$thn + (1-t)h(n-y) = hn - (1-t)hy$, 
we have
$$
\chi^{-s}(hn) - \chi^{-s}(h(n-y)) = \int_0^1g'(t)dt =  h\int_0^1y\cdot\nabla\chi^{-s}(hn - (1-t)hy)dt.
$$
Since $0 < h < 1/2$, we have
\begin{equation}
\begin{split}
1 + |hn- (1 - t)hy| \geq 1 + |hn| - h \geq \frac{1}{2} + |hn| 
\geq \frac{1}{2}( 1 + |hn|).
\end{split}
\nonumber
\end{equation}
Then, 
we have
\begin{equation}
\begin{split}
\big|\chi^{-s}(hn) - \chi^{-s}(h(n-y))\big| & \leq Ch\int_0^1(1 + |hn-(1-t)hy)|)^{-(s+1)}dt \\
& \leq Ch(1 + |hn|)^{-s-1},
\end{split}
\nonumber
\end{equation}
which proves the lemma.
\end{proof}

\begin{lemma} 
\label{Lemma10.6}
For $s \geq 0$, we have
$$
\|\big[D_{h,j},\chi^{-s}_h\big]u\| + \|\big[D_{h,j}^{\ast},\chi^{-s}_h\big]u\| \leq C_s\|\chi_h^{-(s+1)}u\|.
$$
\end{lemma}

\begin{proof}
Use (\ref{Dhjchihcommutator}), (\ref{Dhjastchicommutator}) and Lemma \ref{Lemmachidifference}.
\end{proof}

 By virtue of  Lemma \ref{Lemma10.6},
$[D_{h,j},\chi_h^{-s}]\chi_h^{s+1}$ and $[D_{h,j}^{\ast}\chi_h^{-s}]\chi_h^{s+1}$
are uniformly bounded with respect to $0 < h < h_0$ in $L^2({\bf Z}^d)$.

\begin{lemma}
\label{LemmaEstimatesofwidetildeuhinL2-s}
Let $u_h$ be a solution to the equation
\begin{equation}
(- \Delta_{disc,h} - z)u_h= f_h.
\label{Lemma10.7equation}
\end{equation}
Then for any $s \geq 0$,
\begin{equation}
\|\chi(x)^{-s}\widetilde u_h(x)\|_{H^2({\bf R}^d)} \leq 
C_s\big(\|\widetilde f_h\|_{L^{2,-s}({\bf R}^d)} + \|\widetilde u_h\|_{L^{2,-s}({\bf R}^d)}\big).
\label{AprioriestimateinL2-s}
\end{equation}
\end{lemma}

\begin{proof}
Noting that
$$
C^{-1}|k|\leq \left| \frac{1}{h}\big(1 - e^{ihk}\big)\right| \leq C|k|,
$$
$$
C^{-1}|k|^2\leq \left| \frac{1}{h^2}\big(1 - \cos(hk)\big)\right| \leq C|k|^2,
$$
for $k \in {\bf R}, \ |k| \leq \pi/h$, we have

\begin{equation}
C^{-1}\|\Delta \widetilde u_h\|_{L^2({\bf R}^d)} \leq 
\|\Delta_{disc,h}u\|_{L^2({\bf Z}^d_h)} \leq C\|\Delta \widetilde u_h\|_{L^2({\bf R}^d)}.
\nonumber
\end{equation}
This and (\ref{Lemma10.7equation}) imply  (\ref{AprioriestimateinL2-s}) for $s = 0$. 

Note that we also obtain
\begin{equation}
C\|\xi\widehat u_h(\xi)\|_{L^2({\bf T}^d_h)} \leq \|D_{h,j}u_h\|_{L^2({\bf Z}^d_h)} \leq C^{-1}\|\xi\widehat u_h(\xi)\|_{L^2({\bf T}^d_h)}.
\end{equation}

Letting $g_s(n) = \chi^{-s}_h(n)f_h(n)$ and $v_s = \chi_h^{-s}u_h$, we have
\begin{equation}
( - \Delta_{disc,h} - z)v_s = g_s - \big[\Delta_{disc,h},\chi_h^{-s}\big]u.
\nonumber
\end{equation}
By (\ref{Dhjchihcommutator}) and (\ref{Dhjastchicommutator}), the $L^2({\bf Z}^d_h)$ norm of the right-hand side is estimated from above by
$\|g_s\|_{L^{2}({\bf Z}^d_h)} + \|v_s\|_{L^2({\bf Z}^d_h)}$.
Then  (\ref{AprioriestimateinL2-s}) for $s > 0$ follows from the case $s = 0$.
\end{proof}

\subsection*{Acknowledgements}
The authors were partially supported by the Danish Council for Independent Research $|$ Natural Sciences, Grant
DFF--8021-0084B,  Grants-in-Aid for
Scientific Research (S) 15H05740, and Grants-in-Aid for
Scientific Research (C) 20K03667, Japan Society for the Promotion of Science.

\end{document}